\documentclass{lmcs} 
\pdfoutput=1
\usepackage[utf8]{inputenc}

\usepackage{lastpage}
\lmcsdoi{20}{2}{14}
\lmcsheading{}{\pageref{LastPage}}{}{}%
{Aug.~03,~2023}{Jun.~13,~2024}{}

\usepackage{stmaryrd}
\usetikzlibrary{intersections, calc,patterns}

\newcommand{\1}{\mbox{1}\hspace{-0.25em}\mbox{l}}
\newcommand{\F}{\mathcal{F}}
\newcommand{\E}{\mathbb{E}}
\newcommand{\R}{\mathbb{R}}
\newcommand{\N}{\mathbb{N}}

\newcommand{\Q}{\mathbb{Q}}

\newcommand{\B}{\mathcal{B}}
\newcommand{\Prob}{\mathbb{P}}
\newcommand{\U}{\mathcal{U}}
\newcommand{\bracket}[1]{\llbracket #1\rrbracket}
\newcommand{\bracketo}[1]{\llbracket #1\rrbracket_\omega}
\newcommand{\internal}{\mathrm{int}}
\newcommand{\lrangle}[1]{\langle#1\rangle}
\newcommand{\asure}{\hspace{10pt}\text{a.s. }}

\title[MTL for stochastic processes]{On the Metric Temporal Logic for\texorpdfstring{\\}{} Continuous Stochastic Processes}

\author[M. Ikeda]{Mitsumasa Ikeda\lmcsorcid{0000-0003-1416-7672}}[a]
\address{National Institute of Advanced Industrial Science and Technology (AIST), Osaka 563-8577, Japan}
\email{mikeda.mscmath@gmail.com, yoriyuki.yamagata@aist.go.jp}

\author[Y. Yamagata]{Yoriyuki Yamagata\lmcsorcid{0000-0003-2096-677X}}[a]

\author[T. Kihara]{Takayuki Kihara\lmcsorcid{0000-0002-1611-952X}}[b]
\address{Nagoya University, Nagoya 464-8601, Japan}
\email{kihara@i.nagoya-u.ac.jp}


\begin{document}

\begin{abstract}
In this paper, we prove the measurability of an event for which a general continuous--time stochastic process satisfies the continuous--time Metric Temporal Logic (MTL) formula.
Continuous--time MTL can define temporal constraints for physical systems naturally.
Several previous studies deal with the probability of continuous MTL semantics for stochastic processes.
However, proving measurability for such events is not an obvious task, even though it is essential.
The difficulty comes from the semantics of ``until operator,'' which is defined by the logical sum of uncountably many propositions.
Given the difficulty in proving such an event's measurability using classical measure-theoretic methods, we employ a theorem from stochastic analysis.
This theorem is utilized to prove the measurability of hitting times for stochastic processes, and it stands as a profound result within the theory of capacity.
Next, we provide an example that illustrates the failure of probability approximation when discretizing the continuous semantics of MTL--formulas with respect to time. 
Additionally, we prove that the probability of the discretized semantics converges to that of the continuous semantics when we impose restrictions on diamond operators to prevent nesting.

\end{abstract}

\maketitle

\section{Introduction}

Stochastic processes have emerged as valuable tools for analyzing real-time dynamics characterized by uncertainties. They consist of a family of random variables indexed in real-time and find applications in diverse domains such as molecular behavior, mathematical finance, and turbulence modeling. To formally analyze the temporal properties of real-time systems, \emph{Metric Temporal Logic (MTL)} has been introduced as a logical framework, specifying constraints that real-time systems must satisfy (see Chapter VI in \cite{ouc.200304121619670101}). The increasing demand for MTL specifications in industrial applications~\cite{10.1007-BF01995674,edssjb.978.3.642.24372.1.120110101,edseee.473936620081201} has sparked interest in investigating the probability that a stochastic process satisfies the semantics of MTL--formulas.
In contrast to discrete-time stochastic systems, which are limited to describing events within a discretized time domain, continuous-time MTL allows for the precise representation of constraints on events occurring between discrete times. 

This paper focuses on MTL specifications on stochastic systems interpreted by the continuous-time domain.
In particular, we are interested in the probability of an MTL--formula being satisfied by a continuous--time stochastic system.
However, before discussing the probability of event occurrences, it is crucial to ensure their measurability. 
Although previous papers~\cite{edsdbl.conf.cdc.FuT1520150101,edseee.690764320140501,edseee.758791120160901} considered the probability of events in which an MTL--formula or MITL--formula, a restriction of MTL--formulas is satisfied, they did not prove the measurability of such events.
The subtle problem arises in the definition of probability because unions of uncountably many sets define temporal operators in MTL. At the same time, measurability is guaranteed in the case of a union of countably many sets in general.
The same problem persists even when we restrict our attention to MITL--formulas because MITL--formulas contain intervals in which uncountable instances of time are contained.

In this paper, we prove the measurability of events where sample paths of stochastic processes satisfy the propositions defined by MTL under mild assumptions. 
We assume that the stochastic process is measurable as a mapping from the product space of the time domain and the sample space, a common assumption in stochastic analysis.
Although establishing measurability for events represented by MTL--formulas with temporal operators is made challenging by the presence of the union or intersection of uncountably many sets in the definition, we overcome this difficulty by introducing the concept of reaching time for sub-formulas of the given MTL--formula.
By leveraging the reaching time, we prove measurability inductively concerning the sub-formulas.
In the proof of measurability of MTL--formulas, we utilize the measurability of the corresponding reaching times.
The measurability of such reaching times is non-trivial and proven using the theory of capacity (see \cite{ouc.200324971419720101}).

While establishing measurability shows that the probability of continuous--time MTL semantics for stochastic processes is well-defined, the problem of calculating such probabilities remains a challenge. 
Although \cite{edsdbl.conf.cdc.FuT1520150101,edseee.758791120160901} proposed an approximation by discretizing the semantics of MTL--formulas with respect to time, we show that probability based on discretization does not converge the probability based on continuous semantics in general.
We give an example that involves multi-level nested temporal operators.
The example motivates the need for a more comprehensive and precise discussion of approximations, which has been generally overlooked in previous studies.
As a part of such an effort, we show that if a formula only has simplified temporal operators, and these operators never appear in nested positions, the discretization converges the continuous semantics.

This paper contributes to understanding the probability foundation and approximation for stochastic processes satisfying continuous-time MTL semantics. We investigate the measurability of events, provide proofs under mild assumptions, and explore the possibility of approximation by discretization with respect to time.
Our findings highlight the challenges and emphasize the importance of refining approximation techniques in future studies.

This paper is organized as follows.
In Section~\ref{sec:relatedworks}, we refer to some related works and the novelty of our results compared to previous studies.
In Section~\ref{sec:preliminaries}, we provide a comprehensive exposition of the fundamental concepts that will be utilized extensively in this paper.
These include the definitions of measurability, probability space, stochastic process, Brownian motion, stochastic differential equation, and metric temporal logic.
In Section~\ref{sec:proof_of_measurability}, we prove the measurability of the set of paths of a stochastic process satisfying the semantics of the MTL--formula in both continuous and discrete sense.
In Section~\ref{sec:discretization}, we provide a counterexample that the probability that a stochastic process satisfies the discrete semantics of an MTL--formula does not converge to the probability of path satisfying the semantics of the same MTL--formula in a continuous sense.
On the other hand, in Section~\ref{sec:flat_MTL}, we prove the convergence of probability of discrete semantics for general stochastic differential equations (SDEs) under some restriction on the syntax of propositional formulas.
We set the restriction so that temporal operators do not nest.

In conclusion, we show that the convergence result relies on the depth of nests of temporal operators in an MTL--formula.

\section{Related Works}\label{sec:relatedworks}

Temporal reasoning has been extensively studied (for an overview, see~\cite{sep-logic-temporal}), and it has gained increasing attention due to the growing demands in various industrial applications for real-time systems.

Pnueli~\cite{edsdbl.conf.focs.Pnueli7719770101} introduced linear temporal logic (LTL) as a means to express qualitative timing properties of real-time systems using temporal operators. Koyman~\cite{10.1007-BF01995674} extended this logic to include quantitative specifications by indexing the temporal operators with time intervals, leading to the development of metric temporal logic (MTL). Unlike other extensions of LTL with timing constraints, such as timed propositional temporal logic (TPTL) \cite{edseee.6347319890101}, MTL does not allow explicit reference to clocks, making it practical for implementation. A more detailed survey of temporal logic for real-time systems can be found in ~\cite{edssjs.E69AE06D20130601}.

This paper focuses on MTL for a continuous-time stochastic process with a continuous state space. Such processes are commonly used as probabilistic models to describe phenomena with continuous or intermittent effects caused by environmental noise.
In particular, the process represented by the \emph{stochastic differential equation (SDE)} is widely used for model statistical dynamics, asset prices in mathematical finance \cite{edsnuk.vtls00057597320050101,edsnuk.vtls00216662120040101}, computer network~\cite{abouzeid2000stochastic} and future position of aircraft \cite{edseee.127248520030101}, to named a few~\cite{edseee.473936620081201,edseee.690764320140501,edselc.2-52.0-8504760730120171006}.

Considering the wide range of applications, it is natural to consider the probability that the given stochastic system satisfies properties defined by MTL--formulas or MITL--formulas.
The previous studies~\cite{edsdbl.conf.cdc.FuT1520150101,edseee.758791120160901} already considered the probabilities in which stochastic systems satisfy MTL properties and gave an approximation based on the discretization of time and state spaces.

However, to talk probabilities consistently, we must show the measurability of events under consideration.
The subtle problem arises in the definition of probability because unions of uncountably many sets define temporal operators in MTL. At the same time, measurability is guaranteed for the unions of countably many sets in general.
Further, their approximation by discretization assumed that the timed behavior of stochastic processes satisfies Non--Zenoness, which means that the behavior does not change its value infinitely in finite time.
However, stochastic processes, such as solutions of SDEs, generally do not satisfy  Non--Zenoness assumption because of the inherent ``rough'' properties of stochastic processes, as they are neither smooth nor differentiable everywhere (see, for example, Chapter 2 in \cite{MR1121940}).

In this paper, we prove the measurability of events in which stochastic systems satisfy MTL--formulas interpreted by the continuous--time domain under the mild assumption, using the fundamental theory of stochastic analysis, which is developed to study the approximation of probability measures by describing the structure of classes of sets~\cite{ouc.200324971419720101,OUE.MC0002929219950101}.
Our result guarantees the existence of probability in which stochastic systems satisfy MTL--formulas.

Further, we give examples in which probabilities defined by discretization do not converge to probabilities specified in the continuous-time domain, even if the time interval used for discretization goes to 0.
Our examples show that approximation by discretization, proposed by previous studies, does not generally work.
Our examples involve triple nesting ``always'' and ``possibly''.

On the positive side, we show that if MTL--formulas only have ``always'' and ``possibly'' operators and do not nest, time discretization converges to probability defined in the continuous time domain.
Furia and Rossi~\cite{Furia2010-bl} considered the relation between semantics based on discrete time and continuous time.
Although they do not consider the stochastic system, they also find a correspondence between discrete time semantics and continuous semantics when the formula does not have a nested temporal operator.
Therefore, we use the notation $\flat$MTL for our subsystem of MTL, borrowing the symbol from their paper.

The model checking and satisfiability problems of MTL and other interval temporal logics are considered in the previous works~\cite{Alur1996-im,Della_Monica2011-vk,Ouaknine2005-mm}.
Although our work concerns probability holding an MTL--formula, not its validity, their works share the context with us as they concern formal verification of real-time systems.

\section{Preliminaries}\label{sec:preliminaries}
In this section, we introduce several fundamental concepts discussed throughout this paper.
Let us start with the definitions of measurability and probability space.
When defining an event, it is crucial to ensure that it is measurable to give meaning to its probability.
Once the probability space is defined, we define the product space of two probability spaces.
Next, we define a general stochastic process and its path.
Following that, we introduce the definition of Brownian motion and stochastic differential equation.
These two concepts are fundamental to stochastic analysis and form its core.
Lastly, we introduce the syntax and semantics of MTL--formulas, which are defined for every path of a stochastic process.

\subsection{Measurability and Probability}
This subsection introduces the basic definitions used in the measure theory and probability.
Readers who are familiar with these theories may skip this subsection.
More details are available in \cite{ouc.200337561119660101}.

\begin{defi}[$\sigma$--algebra and Measurable space]\label{defi:sigma_algebra}
    Let $\Omega$ be a set and $\F$ be a family of subsets of $\Omega$, i.e., $\F\subset 2^\Omega$.
    $\F$ is called a $\sigma$--algebra if it satisfies the following three conditions:
    \begin{enumerate}[(i)]    % I changed the list environment to enumerate to keep it consinsent with the other listings
        \item $\Omega\in\F$ and $\emptyset\in\F$.
        \item If $A\in \F$, then $\Omega\setminus A\in\F$. 
        \item If $A_n\in\F$ for $i=1,2,3,\cdots$, then $\bigcup_{i=1}^\infty A_n\in\F$ and $\bigcap_{i=1}^\infty A_n\in\F$
    \end{enumerate}
    If $\F$ is $\sigma$--algebra, $(\Omega,\F)$ is called \emph{a measurable space}.
    If $(\Omega,\F)$ is a measurable space and $A\in F$, we say that $A$ is \emph{$\F$--measurable} or merely \emph{measurable}.
\end{defi}

\begin{defi}[Borel space]
    Let $E$ be a topological space.
    A measurable space is called \emph{Borel space} on $E$, denoted by $(E,\B(E))$, if $\B(E)$ is the smallest $\sigma$--algebra which contains every open set.
    Every set belonging to $\B(E)$ is called \emph{a Borel set}.
\end{defi}

\begin{defi}[Measure space and probability space]\label{defi:measure_space}
    Let $(\Omega,\F)$ be a measurable space.
    A function $\mu:\F\rightarrow [0,\infty]$ is called a \emph{measure} on $(\Omega,\F)$ if $\mu$ satisfies the following two conditions:
    \begin{enumerate}[(i)]
        \item $\mu(\emptyset)=0$.
        \item If $A_i\in\F$ for $i=1,2,3,\cdots$ and $A_i\cap A_j=\emptyset$ for $j\neq i$, then
        \begin{align*}
            \mu(\bigcup_{i=1}^\infty A_i)=\sum_{i=1}^\infty \mu(A_i).
        \end{align*}
    \end{enumerate}
    We call the triplets $(\Omega,\F,\mu)$ \emph{a measure space}.

    Especially, if $\mu(\Omega)=1$, we refer to a measure $\mu$ on $(\Omega,\F)$ as a \emph{probability measure}, and call $(\Omega,\F,\mu)$ a \emph{probability space}.
    We often write a probability measure as $\Prob$.
\end{defi}

\begin{defi}[Lebesgue measure]\label{defi:Leb_meas}
    Let $E=[0,\infty)$ or $\R^n$ endowed with the usual topology and $(E,\B(E))$ be the Borel space on $E$.
    It is well known that there exists a unique measure $\mu$ on $E$ such that
    \begin{align}
        \mu(\prod_{i=1}^n\lrangle{a_i,b_i})=\prod_{i=1}^n(b_i-a_i),
    \end{align}
    for every rectangle $\prod_{i=1}^n\lrangle{a_i,b_i}$ on $E$.
    We call such a measure \emph{Lebesgue measure}.
\end{defi}

\begin{defi}[Complete probability space]
    A probability space $(\Omega,\F,\Prob)$ is said to be \emph{complete} if every subset $G$ of a measurable set $F$ such that $\Prob(F)=0$ is also measurable.
\end{defi}

\begin{defi}[Almost sure]
    Let $(\Omega,\F,\Prob)$ be a probability space, and let $P(\omega)$ be a proposition defined on $\omega\in\Omega$.
    We say that $P(\omega)$ holds \emph{almost surely} if there exists a measurable set $\tilde{\Omega}\in\F$ such that $\Prob(\tilde{\Omega})=1$ and $\tilde{\Omega}\subset\{\omega\in\Omega; P(\omega)\}$.

In the context of probability theory, the phrase ``almost surely'' is often denoted as ``a.s.''.
Hence, we frequently use the notation $P(\omega),\text{ a.s.}$ to indicate that $P(\omega)$ holds almost surely.
\end{defi}

\begin{rem}
    Let $(\Omega,\F, \Prob)$ be a probability space and $P(\omega)$ be a proposition defined on $\omega \in \Omega$.
    Whether $P(\omega)$ holds almost surely depends on the probability measure $\Prob$.
    For example, let $\Omega = [0, 1]$ and $\F := \B ([0, 1])$.
    Let $\Prob$ be the Lebesgue measure on $(\Omega, \F)$.
    Then $\Prob$ is a probability measure on $(\Omega, \F)$.
    Now define another probability measure $\tilde{\Prob}$ as follows:
    \begin{align*}
    \begin{cases}
        \tilde{\Prob}(A) = 1,&\text{ if }x\in A,\\
        \tilde{\Prob}(A) = 0,&\text{ if }x\notin A,
    \end{cases}
    \end{align*}
    where $x \in [0,1]$.
    Define a proposition $P(\omega)$ for $\omega \in \Omega$ as follows:
    \begin{align*}
        P(\omega) \Leftrightarrow \omega=x.
    \end{align*}
    Then $\Prob(\omega \in \Omega; P(\omega)) = 0$, while $\tilde{\Prob}(\omega \in \Omega; P(\omega)) = 1$.
    When we claim that $P(\omega)$ holds almost surely focusing on the probability measure $\Prob$, we express that ``\emph{$P(\omega)$ holds almost surely $\Prob$}'' or ``\emph{$P(\omega)$ holds a.s. $\Prob$}''.
\end{rem}

\begin{rem}\label{rema:almost_sure_concerves}
    Let $P_1(\omega)$ and $P_2(\omega)$ be two proposition defined on $\omega\in\Omega$.
    If $P_1(\omega)$ holds almost surely and $P_2(\omega)$ holds almost surely, then both $P_1(\omega)$ and $P_2(\omega)$ hold almost surely.
    To see this, let $\Omega_1\in\F$ and $\Omega_2\in\F$ and suppose that $\Omega_1\subset \{\omega\in\Omega ; P_1(\omega)\}$, $\Omega_2\subset\{\omega\in\Omega ; P_2(\omega)\}$, and $\Prob(\Omega_1)=\Prob(\Omega_2)=1$. 
    Then $\Omega_1\cap\Omega_2\subset\{\omega\in\Omega ; P_1(\omega)\text{ and } P_2(\omega)\}$ and 
    \begin{align*}
        \Prob((\Omega_1\cap\Omega_2)^C)=\Prob(\Omega_1^C\cup\Omega_2^C)\leq\Prob(\Omega_1^C)+\Prob(\Omega_2^C)-\Prob(\Omega_1^C\cap\Omega_2^C)\leq\Prob(\Omega_1^C)+\Prob(\Omega_2^C)=0.
    \end{align*}
    Therefore, $\Prob(\Omega_1\cap\Omega_2)=1$.

    By repeating similar arguments, we can see the following statement: if $P_i(\omega)$ holds almost surely for $i=1,\cdots,n$, then $P_1(\omega)\wedge,\cdots,\wedge P_n(\omega)$ holds almost surely.
\end{rem}

\begin{defi}[Product measurable space]
    Let $(G,\mathcal{G})$ and $(H,\mathcal{H})$ be two measurable spaces.
    The product $\sigma$--algebra of $\mathcal{G}$ and $\mathcal{H}$, denoted $\mathcal{G}\otimes\mathcal{H}$, is the smallest $\sigma$--algebra on $G\times H$ which contains all set of the form $A\times B$, where $A\in \mathcal{G}$ and $B\in\mathcal{H}$.
    The resulting measurable space
    $(G\times H,\mathcal{G}\otimes\mathcal{H})$ is called the product measurable space of $(G,\mathcal{G})$ and $(H,\mathcal{H})$.
\end{defi}

\begin{fact}\label{fact:section_measurability}
    Let $(G,\mathcal{G})$ and $(H,\mathcal{H})$ me two measurable space and $E\in\mathcal{G}\otimes\mathcal{H}$.
    Then the following holds:
    \begin{align*}
        \{z\in H ; (x,z)\in E\}\in\mathcal{H},\\
        \{z\in G ; (z,y)\in E\} \in \mathcal{G}.
    \end{align*}
\end{fact}

\begin{defi}[Product measure space]
    Let $(\Omega,\F,\Prob)$ be a probability space.
    Let $E=[0,\infty)$ or $\R^n$ endowed with the usual topology, and $(E,\B(E),\mu)$ be the measure space with the Lebesgue measure $\mu$ defined in Definition~\ref{defi:Leb_meas}.
    It is well known that there is a unique measure $\Prob\otimes \mu$ on $(\Omega\times E,\F\otimes\B(E))$ such that
    \begin{align}
        \Prob\otimes\mu(A\times B)=\Prob(A)\times\mu(B)
    \end{align}
    for every $A\in\F$ and $B\in\B(E)$.
    We call such measure \emph{the product measure} of $\Prob$ and $\mu$.
    The resulting measure space is denoted as $(\Omega\times E,\F\otimes\B(E),\Prob\otimes\mu)$.
\end{defi}

\begin{defi}[Measurable function]
    Let $(G,\mathcal{G})$ and $(H,\mathcal{H})$ be two measurable spaces.
    A function $X:G\rightarrow H$ is said to be $\mathcal{G}/\mathcal{H}$--measurable if $X^{-1}(B)\in\mathcal{G}$ for every $B\in\mathcal{H}$.
\end{defi}

\begin{defi}[Random variable]
    Let $(\Omega,\F,\Prob)$ be a probability space and $(E,\B(E))$ be the Borel space on a topological space $E$.
    An $E$--valued function $X$ on $(\Omega,\F,\Prob)$ is called a random variable if $X$ is $\F/\B(E)$--measurable.
\end{defi}
If there is no ambiguity, an $E$-valued random variable is simply called a random variable.

\begin{defi}[Law of a random variable]
    Let $(\Omega,\mathcal{F},\mathbb{P})$ be a probability space, $(E,\mathcal{B}(E))$ be the Borel space on a topological space $E$, and $X$ be an $E$-valued random variable.
    It is well-known that $\sigma(X):=\{X^{-1}(B) ; B\in \mathcal{B}(E)\}$ is a sub $\sigma$--algebra of $\mathcal{F}$, and thus the mapping 
    $B\mapsto \mathbb{P}(X^{-1}(B))$ is a probability measure on $(E,\mathcal{B}(E))$.
    We refer to this mapping as \emph{the law of $X$}, often denoted by $\Prob^X$.
\end{defi}

\begin{defi}[Lebesgue integral and expected value]
    Let $(\R,\B(\R))$ be the Borel space on $\R$.
    Suppose that $(G,\mathcal{G},\mu)$ is a measure space and $f:G\rightarrow \R$ be a $\mathcal{G}/\B(\R)$--measurable function.
    The  \emph{Lebesgue integral} is defined as the following four steps:
    \begin{enumerate}[1.]
        \item Let $B\in\mathcal{G}$ and define $\1_B:G\rightarrow \{0,1\}$ as
        \begin{align}
            \1_B(x):=\begin{cases}
                1,&\text{ if }x\in B\\
                0,&\text{ if }x\notin B.
            \end{cases}
        \end{align}
        We call $f:G\rightarrow [0,\infty)$ \emph{Simple function} if
        \begin{align}
            f=\sum_{i=1}^n\alpha_i \1_{B_i},
        \end{align}
        where $B_i,\ i=1,2,\cdots,n$ are elements in $\mathcal{G}$ and $\alpha_1,\alpha_2,\cdots,\alpha_n$ are nonnegative real numbers.
        Let $A\in\mathcal{G}$.
        Define \emph{the Lebesgue integral of the simple function $f$ with respect to $\mu$} as
        \begin{align}
            \int_A f d\mu:=\sum_{i=1}^n\alpha_i\mu(B_i\cap A).
        \end{align}
        \item Let $A\in\mathcal{G}$ and $f:G\rightarrow\R$ be a nonnegative $\mathcal{G}/\B(\R)$--measurable function.
        \emph{The Lebesgue integral} of $f$ with respect to $\mu$ is defined as follows:
        \begin{align}
            \int_A fd\mu:=\sup\left\{\int_A g d\mu\ ;\ g\text{ is simple function such that }g\leq f \right\}.
        \end{align}
        \item Let $A\in\mathcal{G}$ and $f:G\rightarrow \R$ be a $\mathcal{G}/\B(\R)$--measurable function.
        Let $f^+:=\max\{f,0\}$ and $f^-:=-\min\{f,0\}$.
        It is well known that $f^+$ and $f^-$ are $\mathcal{G}/\B(\R)$--measurable and then $\int_A f^+d\mu$ and $\int_A f^-d\mu$ can be defined.
        When $\int_A f^+d\mu<\infty$ and $\int_A f^-d\mu<\infty$, we define \emph{the Lebesgue integral} with respect to $\mu$ as
        \begin{align}
            \int_A fd\mu:=\int_A f^+d\mu-\int_A f^-d\mu.
        \end{align}
        If $A=G$, then we denote $\int_A fd\mu$ as $\int fd\mu$.
        \item Let $(\Omega,\F,\Prob)$ be a probability space and $X$ be a $\R$--valued random variable such that $\int_\Omega X^+ d\Prob<\infty$ and $\int_\Omega X^-d\Prob<\infty$.
        Then $\int_\Omega Xd\Prob$ is called \emph{the expected value} of $X$, denoted by $\E[X]$.
    \end{enumerate}
    
\end{defi}

\begin{rem}
    In this paper, we use another type of notation for the Lebesgue integral to accommodate different situations:
    \begin{align}
        \int_A fd\mu=\int_A f(x)\mu(dx)
    \end{align}
    Especially, if $\mu$ is a Lebesgue measure, we denote the integral of $x \mapsto f(x)$ as following:
    \begin{align}
        \int_A f(x)dx
    \end{align}
\end{rem}

\begin{defi}[Density of random variable and absolute continuity]
\ 
\begin{enumerate}[(i)]
    \item Let $(\Omega,\F,\Prob)$ be a probability space, $(\R^n,\B(\R^n))$ be a Borel space, and $X$ be an $\R^n$--valued random variable on $(\Omega,\F,\Prob)$.
    Let $([0,\infty),\B([0,\infty)))$ is a Borel space on $[0,\infty)$ endowed with the usual topology.
    A $\B(\R^n )/\B([0,\infty))$--measurable function $f$ is called \emph{the density} of $X$ if the following holds:
    \begin{align*}
        \Prob(X^{-1}(B))=\int_B f(x)dx,\hspace{10pt}\forall B\in \B(\R^n).
    \end{align*}
    If such a function $f$ exists, we say that $X$ has a density.
    \item If $X$ has a density, we say that the law of $X$ is \emph{absolutely continuous} with respect to the Lebesgue measure.
\end{enumerate}
    
\end{defi}

In the following sections, we frequently use the notion of almost sure convergence of random variables:

\begin{defi}[Almost sure convergence]
    Let $(\Omega,\F,\Prob)$ be a probability space, and let $E$ be a metric space.
    Let $X$ and $X_n ; n=1,2,\cdots$ be $E$--valued random variables.
    We say $X_n$ converges almost surely to $X$ if there exists a measurable set $N\in\F$ such that $\Prob(N)=0$ and 
    \begin{align*}
        X_n(\omega)\overset{n\rightarrow\infty}{\longrightarrow}X(\omega)
    \end{align*}
    for every $\omega\notin N$.
\end{defi}

\subsection{Stochastic process}\label{sec:stochastic_process}

\begin{defi}
    Let $(\Omega,\F,\Prob)$ be a probability space and $E$ be a Polish space.
    A family of $E$--valued random variables $X:=\{X_{t}\}_{t\geq0}$ indexed by time parameter $t$ is called a \emph{stochastic process}:
\begin{align}
	\begin{array}{ccc}
		\Omega & \stackrel{X_t}{\longrightarrow} & E \\
		\rotatebox{90}{$\in$} && \rotatebox{90}{$\in$} \\
		\omega & \longmapsto & X_t(\omega)
	\end{array}\nonumber
\end{align}

\end{defi}

Following the convention of stochastic analysis, we say $X$ is \emph{measurable} if it satisfies the following assumption:
\begin{asm}\label{cond:measurability}
The function $(\omega,t)\mapsto X_t(\omega)$ is $\mathcal{F}\otimes\mathcal{B}([0,\infty))$-measurable, which means that the inverse image $\{(\omega,t) ; X_t(\omega)\in B\}$ belongs to $\mathcal{F}\otimes\mathcal{B}([0,\infty))$ whenever $B$ is a Borel set in $E$.
\end{asm}

\begin{rem}
Under Assumption~\ref{cond:measurability}, it follows from Fact~\ref{fact:section_measurability} that $X_t$ is $\mathcal{F}/\B(E)$-measurable for all $t\in[0,\infty)$.
\end{rem}

We denote a path $t\mapsto X_t(\omega)$ of $\{X_t\}_{t\geq0}$ as $X(\omega)$ for every $\omega\in\Omega$:
\begin{align}
		\begin{array}{ccc}
			[0,\infty) & \stackrel{X(\omega)}{\longrightarrow} & E \\
			\rotatebox{90}{$\in$} && \rotatebox{90}{$\in$} \\
			t & \longmapsto & X_t(\omega)
		\end{array}\nonumber
\end{align}

\begin{rem}\label{rem:conti_measurable}
    Suppose that $(\Omega,\F,\Prob)$ is a complete probability space.
    If $X_t$ is $\mathcal{F}$-measurable for all $t\in[0,\infty)$ and one of the following holds, $X$ is known to be measurable (see 1.1.14 in \cite{MR1121940}):
    \begin{enumerate}[(i)]
        \item The path $t \mapsto X_t$ is right--continuous almost surely $\Prob$.
        \item The path $t \mapsto X_t$ is left--continuous almost surely $\Prob$.
    \end{enumerate}
\end{rem}

\subsection{Brownian motion}
When studying properties of distributions of general continuous stochastic processes and topics such as convergence of discretization, including numerical computations, it is common to first discuss examples related to a Brownian motion as it is the most representative continuous stochastic process.
Now, we present the definition of the Brownian motion: 

\begin{defi}\label{defi:BM}
    Let $(\Omega,\F,\Prob)$ be a probability space.
    A stochastic process $X:=\{X_t\}_{t\geq0}$ with state space $\R$ is called a standard one-dimensional Brownian motion starting at $x\in\R$ if 
	\begin{enumerate}[(i)]
		\item $\Prob(X_0=x)=1$,
		  \item The path $t\mapsto X_t$ is continuous almost surely $\Prob$,
		  \item For any $s,t\geq0$, $t>s$ implies that $X_t-X_s\sim\mathcal{N}(0,t-s)$ i.e., $X_t-X_s$ has normal distribution with mean $0$ and variance $t-s$.
             \item If $s\leq t\leq u$, then $X_u-X_t$ is independent of $\sigma(X_v;v\leq s)$, where $\sigma(X_v;v\leq s)$ is the smallest sigma algebra which contains $\sigma(X_v)$ for all $v\leq s$.
	\end{enumerate}
\end{defi}

The existence of a Brownian motion is established in Section 2.2 of \cite{MR1121940}, relying on the richness of the underlying probability space.

 \subsection{Stochastic differential equation (SDE)}    % I added a space before "(SDE)". If you are not ok with this, let us know
Let us proceed to define one-dimensional stochastic differential equations rigorously:

\begin{defi}\label{defi:SDE_with_drift}
    Let $(\Omega,\F,\Prob)$ be a probability space that supports a Brownian motion $W$.
    Let $\sigma$ and $b$ be real--valued measurable functions on $\R$.
    A {\bf strong solution of the stochastic differential equation (SDE)} 
    \begin{align*}
\begin{cases}
dX_t=b(X_t)dt+\sigma(X_t)dW_t,\\
X_0=\xi\in\R.
\end{cases}
\end{align*} 
    on $(\Omega,\F,\Prob)$ with respect to $W$ and initial condition $\xi$, is a process $X=\{X_t\}_{t\geq0}$ with continuous sample paths and with the following properties:
    \begin{enumerate}[(i)]
	\item $\{X_t\}_{t\geq0}$ is adapted to the filtration induced by Brownian motion (see 1.1.9 and 5.2.1 in \cite{MR1121940}),
        \item $\Prob[\omega\in\Omega; X_0(\omega)=\xi]=1$,
	\item $\Prob[\omega\in\Omega;\int_0^t\{|b(X_s(\omega))|+\sigma^2(X_t(\omega))\}ds<\infty]=1$ holds for every $0\leq t<\infty$, and
	\item the integral version of \eqref{eq:SDE_with_drift}
            \begin{align*}
			X_t=X_0+\int_0^tb(X_s)ds+\int_0^t\sigma(X_s)dW_s;\hspace{10pt}0\leq t<\infty,
		\end{align*}
	holds almost surely. 
    Here, the term $\int_0^t \sigma(X_s) dW_s$ represents the {\it Itô's stochastic integral}, defined as the limit of the following stochastic process (refer to Chapter 3 in \cite{MR1121940}):
     \begin{align}
        \sum_{k=1}^\infty \sigma(X_{\frac{k}{n}})(W_{\frac{k+1}{n}\wedge t}-W_{\frac{k}{n}\wedge t}).
    \end{align}
    \end{enumerate}
\end{defi}

\begin{rem}\label{rem:Brownian_as_SDE}
    Brownian motion $\{W_t\}_{t\geq0}$ starting at $x\in\R$ itself is a solution $\{X_t\}_{t\geq0}$ of following one--dimensional SDE:
    \begin{align*}
    \begin{cases}
        X_t=\int_0^t1dW_s,\\
        X_0=x.
    \end{cases}
    \end{align*}
\end{rem}

\begin{rem}
    Brownian motions and SDEs are well-known to be continuous stochastic processes.
    Therefore, they satisfy Assumption~\ref{cond:measurability}.
\end{rem}

\subsection{Metric Temporal Logic (MTL) and its semantics}

 This section introduces the formal definition of MTL (Metric Temporal Logic) formulas for a given path.
 We begin by assuming a set of atomic propositions, denoted as $AP$, typically defined as a finite set.
 The MTL--formulas are then defined as follows:

 \begin{defi}[Syntax of MTL--formulas]\label{defi:syntax_MTL}
	We define MTL--formulas for a continuous stochastic process $\{X_t\}_{t\geq0}$ using the following grammar:

 \begin{enumerate}
     \item Every atomic proposition $a\in AP$ is an MTL--formula.
     \item If $\phi$ is an MTL--formula, then $\lnot \phi$ is also an MTL--formula.
	 \item If $\phi_1$ and $\phi_2$ are MTL-formulas, then the conjunction of $\phi_1$ and $\phi_2$, denoted as $\phi_1\wedge\phi_2$, is also an MTL--formula.
  \item If 
  $\phi_1$ and 
  $\phi_2$ are MTL--formulas, and 
  $I$ represents an interval on the domain 
  $[0,\infty)$, then the formula
  $\phi_1\U_I\phi_2$ is an MTL--formula.
 \end{enumerate}
 
	The grammar above can be conveniently represented using the \emph{Backus--Naur form}:
	\begin{align*}
	 \phi::=a \mid \phi_1\wedge\phi_2\mid \lnot\phi\mid\phi_1\U_I\phi_2,
 \end{align*}	 
 \end{defi}

\begin{rem}[``Until'' operator]
    In (3) in the definition~\ref{defi:syntax_MTL}, $\U_I$ is called an \emph{until operator.}
    The interval $I$ appearing in the until operator $\U_I$ can be closed, left open, right open, or purely open. This means that $I$ can take the form $I=[a,b]$, $(a,b]$, $[a,b)$, or $(a,b)$, respectively. Furthermore, when $I$ is unbounded, it can only be of the form $I=(a,b)$ or $I=[a,b)$, where $b$ can take the value $\infty$.	
\end{rem}

We define two types of semantics for the previously presented syntax: one for the continuous time domain and the other for the discrete--time domain.

\begin{defi}[Continuous Semantics of MTL-Formulas]\label{defi:contiMTL}
	Consider a path $X(\omega)$ of the stochastic process $\{X\}_{ t\geq0}$ with a fixed $\omega\in\Omega$. Additionally, for each atomic proposition $a\in\mathrm{AP}$, let us assign a Borel set $B_a$ on the domain $E$. The continuous semantics of MTL--formulas are recursively defined as follows:
	\begin{align*}
		X(\omega),t\models a&\iff&&X_t(\omega)\in B_a\\
		X(\omega),t\models \lnot \phi &\iff&&\text{not }[X(\omega),t\models\phi]\\
		X(\omega),t\models \phi_1\wedge\phi_2&\iff&&X(\omega),t\models\phi_1\text{ and } X(\omega),t\models\phi_2\\
		X(\omega),t\models \phi_1\U_I\phi_2&\iff&&\exists s\in I \text{ s.t.: } X(\omega),t+s\models \phi_2\text{ and }\\
		&&&\forall s'\in[t,t+s),\ X(\omega),s'\models\phi_1
	\end{align*}
\end{defi}

\begin{defi}[Time set]\label{defi:time_set}
	We introduce the notations $\llbracket\phi\rrbracket$, $\llbracket\phi\rrbracket(t)$, and $\bracketo{\phi}$ as follows:
	\begin{align*}
		\bracket{\phi}&:=\{(\omega,t) ; X(\omega),t\models \phi\},\\
		\bracket{\phi}(t)&:=\{\omega ; X(\omega),t\models \phi\},\\
		\bracketo{\phi}&:=\{t\geq0 ;X(\omega),t\models \phi\}.
	\end{align*}
	In particular, we refer to $\bracketo{\phi}$ as the ``time set'' associated with $\phi$.	
\end{defi}

\begin{defi}[Discrete Semantics of MTL-Formulas]
    Let us consider the path $X(\omega)$ of $\{X_t\}_{t\geq0}$ and the assignment $B_a$ for an atomic proposition $a\in\mathrm{AP}$, as well as Definition~\ref{defi:contiMTL}.
    For any $n\in\N$, we denote $\{k/n; k\in\N\}$ as $\N/n$. The discrete semantics of MTL--formulas for any $t\in\N/n$ is defined recursively as follows:
	\begin{align*}
		X(\omega),t\models_n a&\iff&&X_t(\omega)\in B_a\\
		X(\omega),t\models_n \lnot \phi &\iff&& \text{not }[X(\omega),t\models_n\phi]\\
		X(\omega),t\models_n \phi_1\wedge\phi_2&\iff&&X(\omega),t\models_n\phi_1\text{ and } X(\omega),t\models_n\phi_2\\
		X(\omega),t\models_n \phi_1\U_I\phi_2&\iff&&\exists s\in I\cap \N/n \text{ s.t.: } X(\omega),t+s\models_n \phi_2\text{ and }\\
		&&&\forall s'\in[t,t+s)\cap \N/n,\ X(\omega),s'\models_n\phi_1
	\end{align*}
\end{defi}

For $t\in\N/n$, we denote by $\llbracket\phi\rrbracket_n(t)$ the set $\{\omega;X(\omega),t\models_n\phi\}$.

\begin{rem}[Constants, Disjunction, Diamond operator and Box Operator]\label{rem:diamond_box}
    We often use two constants of propositional logic \emph{$\top$ (top) } and \emph{$\bot$ (bottom)}.
    Top means undoubted tautology whose truth nobody could ever question, while bottom means undoubted contradiction.
    These two constants can be represented as
    \begin{align*}
        \top=\phi\vee \lnot \phi\\
        \bot=\phi \wedge \lnot \phi
    \end{align*}
    by arbitrary MTL--formula $\phi$.
    We often use the following notation:
    \begin{align*}
        \phi_1\vee \phi_2&=\lnot((\lnot\phi_1)\wedge(\lnot\phi_2)),\\
        \Diamond_I\phi&=\top \U_I\phi,\\
        \Box_I\phi&=\lnot(\Diamond_I\lnot \phi),
    \end{align*}
    We refer to $\Diamond_I$ and $\Box_I$ as the diamond and box operators, respectively.
    In the continuous and discrete semantics, the following equivalences hold:
    \begin{align*}
        X(\omega),t\models\Diamond_I\phi\Leftrightarrow&(\exists s\in I)[X(\omega),t+s\models\phi],\\
        X(\omega),t\models\Box_I\phi\Leftrightarrow&(\forall s\in I)[X(\omega),t+s\models\phi].
    \end{align*}
    \begin{align*}
        X(\omega),t\models_n\Diamond_I\phi\Leftrightarrow&(\exists s\in I\cap \N/n)[X(\omega),t+s\models_n\phi],\\
        X(\omega),t\models_n\Box_I\phi\Leftrightarrow&(\forall s\in I\cap\N/n)[X(\omega),t+s\models_n\phi].
    \end{align*}
\end{rem}

\section{Proof of Measurability}
\label{sec:proof_of_measurability}

As introduced in Definition~\ref{defi:measure_space}, the probability $\Prob(F)$ can only be defined for a measurable set $F$.
Therefore, in order to define $\Prob(\omega\in\Omega ; X(\omega),t\models\phi)$ or $\Prob(\omega\in\Omega;X(\omega),t\models_n\phi)$, it is necessary to show the measurability of $\bracket{\phi}(t)=\{\omega\in\Omega ; X(\omega),t\models\phi\}$ or $\bracket{\phi}_n(t)=\{\omega\in\Omega;X(\omega),t\models_n\phi\}$, respectively.
Since the definition of the discrete semantics of MTL involves the intersection or union of at most countably many sets, the measurability of $\llbracket\phi\rrbracket_n(t) $ follows directly from Definition~\ref{defi:sigma_algebra} of the $\sigma$-algebra.
Then $\Prob(\omega\in\Omega ; X(\omega),t\models_n\phi)$ can be defined.
However, the measurability of $\llbracket\phi\rrbracket(t)$ is not straightforward.
Then, it is not apparent whether $\Prob(\omega\in\Omega; X(\omega),t\models\phi)$ can be defined.

To illustrate the difficulty, let $X$ be an $E$--valued stochastic process, $a,b$ be atomic propositions, and $I$ be an interval on $[0,\infty)$.
Then $X(\omega),t\models a$ is equivalent to $X_t(\omega)\in B_a$ for some Borel set $B_a$, and $X(\omega),t\models b$ is equivalent to $X_t(\omega)\in B_b$ for some Borel set $B_b$.
Since $X_t$ is $\F/\B(E)$--measurable, then $\bracket{a}(t)=\{\omega\in\Omega ; X_t(\omega)\in B_a\}$ belongs to $\F$ and hence $\Prob(\omega\in\Omega; X(\omega),t\models a)$ can be defined.
$\Prob(\omega\in\Omega ; X(\omega),t\models b)$ can be defined for the same reason.

However, can we define $\Prob(\omega\in\Omega; X(\omega),t\models a\U_Ib)$?
From the definition of the until operator, $X(\omega),t\models a\U_I b$ is equivalent to the following:

\begin{align*}
    (\exists s\in I)[X_{t+s}(\omega)\in B_b \text{ and }(\forall s'\in[0,s))[X_{t+s'}\in B_a]].
\end{align*}
Therefore, 

\begin{align*}
    \bracket{a\U_Ib}(t)=\bigcup_{s\in I}\bigcap_{s'\in[0,s)}\{\omega\in\Omega ; X_{t+s}(\omega)\in B_b\}\cap\{\omega\in\omega ;X_{t+s'}(\omega)\in B_a\}.
\end{align*}
Although the measurability of $\{\omega \in \Omega : X_{t+s}(\omega) \in B_b\} \cap \{\omega \in \Omega : X_{t+s'}(\omega) \in B_a\}$ follows from the $\mathcal{F}/\mathcal{B}(E)$-measurability of $X_{t+s}$ and $X_{t+s'}$, the representation of $\llbracket a \U_I b\rrbracket(t)$ involves uncountable intersections and unions of these sets.
Since Definition~\ref{defi:sigma_algebra} guarantees that measurability is preserved under countable unions or intersections, the semantics of negation "$\lnot$" and logical conjunction "$\wedge$" inherits the measurability.
On the other hand, the measurability of $\llbracket a \U_I b\rrbracket(t)=\{\omega\in\Omega;X(\omega),t\models a\U_I b\}$ is not obvious.
Thus, we have observed that the challenge arises when showing the preservation of the measurability under until operator $\U_I$.

Toward this goal, we utilize a profound theorem from the theory of capacity.
We obtain a representation of the until operator $\U_I$ using the inverse image of the \emph{reaching time} (or \emph{debut}) of a set $B$ on $\Omega\times[0,\infty]$.
Then the inverse image is represented as the projection of a measurable set on $\Omega\times[0,\infty)$ to $\Omega$.
By employing capacity theory and the representation of the until operator by the reaching times, we show that the measurability saves under the until operator.

Now let us introduce reaching time:
\begin{defi}
	Consider a subset $B$ of $\Omega\times[0,\infty)$. \emph{The reaching time} or \emph{debut} $\tau_B(\omega,t)$ of $B$ is defined for each $\omega\in\Omega$ as the first time $s>t$ at which $(\omega,s)$ reaches $B$, given by:
	\begin{align*}
				 \tau_B(\omega,t):=\inf\{s>t;(\omega,s)\in B\},
	\end{align*}
	where $\tau_B(\omega,t):=\infty$ if $\{s>t;(\omega,s)\in B \}=\emptyset$.
\end{defi}

\begin{lem}\label{lem:tau_rcontinous}
    Let $B$ be a subset of $\Omega\times[0,\infty)$.
    Then, the reaching time $t\mapsto\tau_B(\omega,t)$ is right-continuous for every $\omega \in \Omega$.	
\end{lem}

\begin{proof}
    Assume $\tau_B(\omega,t)>t$.
    We can express $\tau_B(\omega,t)$ as $t+\alpha$ for some $\alpha > 0$. According to the definition of $\tau_B(\omega,t)$, for every $s$ in the interval $(t,t+\alpha)$, it holds that $(\omega,s) \notin B$. Therefore, we have $\tau_B(\omega,s) = t+\alpha$ for such $s$, and as a result, $\lim_{s\downarrow t}\tau_B(\omega,s) = t+\alpha = \tau_B(\omega,t)$.
	
    Assume $\tau_B(\omega,t)=t$. For every $\epsilon>0$, there exists $\delta\in(0,\epsilon)$ such that $(\omega,t+\delta)\in B$. Therefore, $\tau_B(\omega,s)\leq t+\delta<t+\epsilon$ for every $s\in (t,t+\delta)$, which implies $\lim_{s\downarrow t}\tau_B(\omega,s)=t=\tau_B(\omega,t)$.
\end{proof}

The following lemma is an abstract version of Proposition 1.1.13 in \cite{MR1121940}.

\begin{lem}\label{lemm:2}
    Let $(\Omega, \F, \Prob)$ be a complete probability space.
    Suppose that a $[0,\infty]$--valued stochastic process $\{Y_t\}_{t\geq0}$ satisfies the following:
    \begin{itemize}
        \item $Y_t$ is $\F/\B([0,\infty])$--measurable for every $t \in [0, \infty)$, i.e., $Y_t^{-1}(B) \in \F$ for every $B \in \B([0,\infty])$.
        \item $t \mapsto Y_t(\omega)$ is right continuous for every $\omega \in \Omega$.
    \end{itemize}
    Then, $\{(\omega, t); Y_t(\omega) \in B\} \in \F\otimes\B([0,\infty))$ for every $B\in \B([0,\infty])$ holds.
    In other word, $(\omega,t)\mapsto Y_t(\omega)$ is $\F\otimes\B([0,\infty))/\B([0,\infty])$-measurable.
\end{lem}
\begin{proof}
    For $t>0$, $n\geq1$, and $k=0,1,\ldots$, we define $Y_t^{(n)}(\omega)=Y_{(k+1)/2^n}(\omega)$ for $\frac{k}{2^n}<t\leq\frac{k+1}{2^n}$, and $Y^{(n)}_0=Y_0(\omega)$.
    The mapping $(\omega,t)\mapsto Y^{(n)}_t(\omega)$ from $\Omega\times[0,\infty)$ to $[0,\infty]$ is demonstrably $\F\otimes\B([0,\infty))/\B([0,\infty])$-measurable.
    Furthermore, by right--continuity, we have $Y^{(n)}_t(\omega) \rightarrow Y_t(\omega)$ if $n \rightarrow \infty$ for any $(\omega,t)\in[0,\infty)\times\Omega$.
    Consequently, the limit mapping $(\omega,t)\mapsto Y_t(\omega)$ is also $\F\otimes\B([0,\infty))/\B([0,\infty])$-measurable.	
\end{proof}

Now, let us show the measurability of the reaching time when considering it as a stochastic process.
This result is derived from capacity theory, which guarantees the measurability of the projection (of a well-behaved set).

\begin{lem}
    If $B$ belongs to $\F\otimes\B([0,\infty))$, the mapping $(\omega,t)\mapsto \tau_B(\omega,t)$ is $\F\otimes\B([0,\infty))/\B([0,\infty])$-measurable, i.e., $\{(\omega, t); \tau_B(\omega, t)\in A\}\in \F\otimes\B([0,\infty))$ for all $A \in \B([0,\infty])$.
\end{lem}
\begin{proof}
    From Lemma~\ref{lem:tau_rcontinous}, the reaching times $\tau_B(t,\omega):=\inf \{s > t;(\omega,s)\in B \}$ are right-continuous with respect to $t\in[0,\infty)$.
    Then it is enough to show that $\omega\mapsto \tau_B(\omega,t)$ is $\F/\B([0,\infty])$--measurable because of Lemma~\ref{lemm:2}.
    From the definition of $\tau_B(\omega,t)$, we can represent $\{ \omega ; \tau_B(\omega,t)<u \}$ by using  projection mapping $\pi:\Omega\times[0,\infty]\rightarrow\Omega$ as 
    \begin{align*}
        \{ \omega ; \tau_B(\omega,t)<u \}=\pi(B\cap\{\Omega\times(t,u)\}),\hspace{10pt}\forall u\in[0,\infty].
    \end{align*}
    Since $[0,\infty]$ is a locally compact space with countable basis, $\F$ is complete, and the set $B\cap\{\Omega\times(t,u)\}$ belongs to $\F\otimes \B([0,\infty))$, we can apply \emph{Theorem I-4.14 in~\cite{OUE.MC0002973019910101}} to show $\pi(B\cap\{\Omega\times(t,u)\})\in\F$.
    Therefore, $\{\omega\in\Omega ; \tau_B(\omega,t)<u\}\in\F$ for all $u\geq0$, which implies the $\F/\B([0,\infty])$--measurability of the map $\omega\mapsto\tau_B(\omega,t)$.
\end{proof}

The subsequent lemma regarding the two types of subsets follows directly from basic arguments in measure theory.

\begin{lem}\label{lem::measurableelements}
Let $B\in \F\otimes\B([0,\infty))$.
    Let $f$ and $g$ be functions from $\Omega\times[0,\infty)$ to $[0,\infty]$, which are $\F\otimes\B([0,\infty))/\B([0,\infty])$-measurable.
    Then, the following sets belong to $\F\otimes\B([0,\infty))$.
    \begin{align}
        &\{ (\omega, t) \in \Omega\times[0,\infty); f(\omega,t)\geq g(\omega,t)\},\label{eq:ineq_fg}\\
        &\{(\omega,t)\in\Omega\times[0,\infty);(\omega,f(\omega,t)) \in B\}.
    \end{align}
\end{lem}
\begin{proof}
    Since $f$ and $g$ are $\F\otimes\B([0,\infty))/\B([0,\infty])$--measurable, the set \eqref{eq:ineq_fg} is $\F\otimes\B([0,\infty))$--measurable.

	Because $f$ is a $\F\otimes\B([0,\infty))/\B([0,\infty])$ measurable function, $\tilde{f}(\omega, t) = (\omega, f(\omega, t))$ is $\F\otimes\B([0,\infty))/\F\otimes\B([0,\infty])$--measurable function.
	By considering $B$ as a subset of $\Omega\times[0,\infty]$, it becomes $\F\otimes\B([0,\infty])$--measurable. Consequently, $\tilde{f}^{-1}(B)$ is $\F\otimes\B([0,\infty))$--measurable.
	Because 
	\[
		\{(\omega,t)\in\Omega\times[0,\infty);(\omega,f(\omega,t)) \in B\} = \tilde{f}^{-1}(B)\in\F\otimes\B([0,\infty)),
	\]
	The lemma holds.    
\end{proof}
	
Now, we can prove the measurability of $\llbracket \phi\rrbracket$ and $\llbracket \phi\rrbracket(t)$.

\begin{lem}\label{lemm:meas_until}
	Let $\phi_1$ and $\phi_2$ be two MTL-formulas and suppose that both $\llbracket\phi_1\rrbracket$ and $\llbracket\phi_2\rrbracket$ are in $\F\otimes\B([0,\infty))$.
	Then $\{(\omega,t);X(\omega),t\models\phi_1\U_I\phi_2\}$ belongs to $\F\otimes\B([0,\infty))$.
\end{lem}

We prove Lemma~\ref{lemm:meas_until} using the right continuity of $\tau_i(t),\ i=1,2$ and Lemma~\ref{lemm:2}.
\begin{proof}
    To prove this lemma, we put 
    \begin{align*}
	\tau_1(\omega, t)&:=\tau_{\llbracket\phi_1\rrbracket^C}(\omega, t)=\inf\{s>t;X(\omega),s\not\models\phi_1\},\\	
	\tau_2(\omega, t)&:=\tau_{\llbracket\phi_2\rrbracket}(\omega, t)=\inf\{s>t;X(\omega),s\models\phi_2\}.
    \end{align*}
    By Lemma~\ref{lemm:2}, $\tau_1$ and $\tau_2$ are $\F\otimes\B([0,\infty))/\B([0,\infty])$-measurable.

    We only prove the case where $I=[a,b]$.
    The cases for other forms of $I$ can be proved in similar ways.
    For simplicity, suppose that $a>0$.
    Then $X(\omega),t\models\phi_1\U_I\phi_2$ holds if and only if $X(\omega), t \models \phi_1$ holds and one of the following possibilities holds:
    \begin{enumerate}
        \item $X(\omega),t+a\models\phi_2$ holds and $\tau_1(\omega,t)\geq t+a$
	\item $X(\omega), t+b \models \phi_2$ holds and $\tau_1(\omega, t) \geq t+b$ holds
	\item $\tau_2(\omega, t+a) < t+b$, $X(\omega), \tau_2(\omega, t+a) \models \phi_2$ and $\tau_1(\omega, t) \geq \tau_2(\omega, t+a)$ hold
	\item $\tau_2(\omega, t+a) < t+b$, $X(\omega), \tau_2(\omega, t+a) \not\models \phi_2$ and $\tau_1(\omega, t) > \tau_2(\omega, t+a)$ hold
    \end{enumerate}
    By $\F\otimes\B([0,\infty))/\B([0,\infty])$--measurability of $\tau_1$ and $\tau_2$,
    \begin{align*}
        &\{ (\omega, t) ;\tau_1(\omega,t) \geq t+a\},\\
        &\{ (\omega, t) ;\tau_1(\omega,t) \geq t+b\},\text{ and}\\
        &\{(\omega,t) ; \tau_2(\omega, t+a) < t+b\}
    \end{align*}
    are in $\F\otimes\B([0,\infty))$.
    Thanks to Lemma~\ref{lem::measurableelements},
    \begin{align*}
        &\{(\omega, t) ; \tau_1(\omega, t) \geq \tau_2(\omega, t+a)\}\text{ and}\\
        &\{(\omega, t) ; \tau_1(\omega, t) > \tau_2(\omega, t+a)\}
    \end{align*}
    are in $\F\otimes\B([0,\infty))$.
    Since $\tau_2(\omega,t+a)$ is $\F\otimes\B([0,\infty))/\B([0,\infty])$--measurable and then
    \begin{align*}
        \{ (\omega, t) ; X(\omega), \tau_2(\omega, t+a) \models \phi_2 \}&=\{(\omega,t); (\omega,\tau_2(\omega,t+a))\in\bracket{\phi_2}\},\\
        \{ (\omega, t) ; X(\omega), \tau_2(\omega, t+a) \not\models \phi_2 \}&=\{(\omega,t); (\omega,\tau_2(\omega,t+a))\notin\bracket{\phi_2}\}.
    \end{align*}
    From Lemma~\ref{lem::measurableelements}, both sets are in $\F\otimes\B([0,\infty))$.
    This completes the proof of $\F\otimes\B([0,\infty))$-measurability of $X(\omega),t\models\phi_1\U_I\phi_2$.
\end{proof}

\begin{thm}
	For each MTL-formula $\phi$, $\bracket{\phi}$ is $\F\otimes\B([0,\infty)) $-measurable and $\bracket{\phi}(t)$ is $\F$-measurable for all $t\geq0$.
\end{thm}

\begin{proof}
    We can prove the measurability of $\llbracket \phi \rrbracket$ by induction on $\phi$.
    \begin{itemize}
        \item Atomic Formula: If $\phi$ is an atomic formula, then $\llbracket \phi \rrbracket$ belongs to $\F \otimes \B([0,\infty))$ because the mapping $(\omega,t) \mapsto X_t(\omega)$ is $\F \otimes \B([0,\infty))/\B(E)$-measurable.
        \item Negation: If $\llbracket \phi \rrbracket$ belongs to $\F \otimes \B([0,\infty))$, then $\llbracket \lnot \phi \rrbracket = \bracket{\phi}^C$ clearly belongs to $\F \otimes \B([0,\infty))$.
        \item Conjunction: Suppose $\phi_1$ and $\phi_2$ are two MTL--formulas, and $\llbracket \phi_i \rrbracket$ is $\F \otimes \B([0,\infty))$-measurable for $i=1,2$. Then it is straightforward to show that $\llbracket \phi_1 \wedge \phi_2 \rrbracket = \bracket{\phi_1} \cap\bracket{\phi_2}$ is also $\F \otimes \B([0,\infty))$-measurable.
        \item Until Operator: From Lemma~\ref{lemm:meas_until}, we can obtain $\F \otimes \B([0,\infty))$-measurability of $\llbracket \phi_1 \U_I \phi_2 \rrbracket$.
    \end{itemize}
    Once we have shown that $\bracket{\phi}$ belongs to $\F \otimes \B([0,\infty))$, the fact that $\llbracket \phi \rrbracket(t) \in \F$ follows from Fact~\ref{fact:section_measurability}.
\end{proof}

Since the domain of $\Prob$ is $\F$, we can define $\Prob(\omega ; X(\omega),t\models\phi)$ for all $\phi$ and $t\in[0,\infty)$.

\section{Discretization of MTL--formula: Counterexample}\label{sec:discretization}

Fu and Ufuk~\cite{edsdbl.conf.cdc.FuT1520150101} proposed a methodology for approximating the probability that the solution of a controlled stochastic differential equation (SDE) satisfies an MITL--formula. Their approach involves discretizing both the time and state space of the SDE.
They derive probabilities based on this discretized semantics by utilizing a reachability problem for a timed automaton generated by the SDE and MITL--formula.

The authors argue that their simulation's convergence results from the convergence in distribution of the approximated SDE, whose state space has been discretized. They claim that the probability obtained from the discretized approach converges to the probability derived from the continuous--time semantics of the original SDE.

However, here, we demonstrate that for a one-dimensional Brownian motion, denoted as X, the probability obtained using the discretized semantics does not necessarily converge to the probability obtained using continuous semantics. 
This failure arises because a reaching time of the Brownian motion may have positive density.

It is worth noting that Brownian motion can be viewed as the solution of a stochastic differential equation (SDE) (see Remark~\ref{rem:Brownian_as_SDE}).
Furthermore, every SDE without control can be regarded as a special case of controlled SDEs.
Consequently, Brownian motion can be an illustrative example of a solution to a controlled SDE. Hence, our counterexample aligns with the scenario presented in \cite{edsdbl.conf.cdc.FuT1520150101}.
Consider that $X$ is a one-dimensional Brownian motion starting from $0$.

The one of the ideas of our construction is that, we can define a MTL--formula depending on a behavior of $X$ in a single time instance by nested temporal operators on open intervals.
In fact, our example falls under MITL--formulas, because there is no temporal operators defined on a single instance of time.

Let $p$ be an atomic formula, $B_p:=[1,\infty)$ be the set associated with $p$ and $\tau_p(\omega):=\inf\{t\geq0 ;X_t\in B_p\}$.
In other words, $X(\omega),t\models p\Leftrightarrow X_t(\omega)\geq1$ and $\tau_p(\omega)=\inf\{t\geq0;X(\omega),t\models p\}$.

Put 
	\begin{align}
		\phi_1&:=\Box_{(1,2)}(\Diamond_{(1,4)}p\wedge\lnot\Diamond_{(1,3)}p)\\
		\phi_2&:=(\Diamond_{(1,3)}\phi_1)\wedge(\lnot\Diamond_{(1,2)}\phi_1)\wedge(\lnot\Diamond_{(2,3)}\phi_1),\\
		\phi_3&:=\Diamond_{(1,2)}\phi_2,\\
		\psi&:=(\lnot p)\wedge(\lnot\Diamond_{(0,8)}p)\wedge\phi_3.\label{eq:MTL_psi}
	\end{align}
In one line,

\begin{align}
	\psi:=(\lnot p)\wedge(\lnot\Diamond_{(0,8)}p)\wedge\Diamond_{(1,2)}[(&\Diamond_{(1,3)}\Box_{(1,2)}(\Diamond_{(1,4)}p\wedge\lnot\Diamond_{(1,3)}p))\\
	\wedge&(\lnot\Diamond_{(1,2)}\Box_{(1,2)}(\Diamond_{(1,4)}p\wedge\lnot\Diamond_{(1,3)}p))\\
	\wedge&(\lnot\Diamond_{(2,3)}\Box_{(1,2)}(\Diamond_{(1,4)}p\wedge\lnot\Diamond_{(1,3)}p))].
\end{align}

In this setting, the following statements hold.
\begin{thmC}[{\cite[Remark 2.8.3]{MR1121940}}]\label{fact:hitting_time_density}
$\tau_p(\omega)$ has positive density on $[0,\infty)$.
\end{thmC}

We will take $\psi$ as a counterexample that $\Prob( \omega \in \Omega ; X(\omega), 0 \models_n \psi ) $ does not converges to $\Prob( \omega \in \Omega ; X(\omega), 0 \models \psi )$.
We show that $\Prob( \omega \in \Omega ; X(\omega), 0 \models_n \psi ) = 0$ for sufficiently large $n\in \N$, while $\Prob( \omega \in \Omega ; X(\omega), 0 \models \psi )>0$. 

Before showing that $\psi$ is the counterexample, let us describe the meaning of the formula $\psi$.
Note that the following equivalences hold:
\begin{align*}
    X(\omega), 0 \models \psi &\Leftrightarrow [X(\omega), 0 \models (\lnot p) \wedge (\lnot \Diamond_{(0,8)} p)] \text{ and }[X(\omega), 0 \models \phi_3],\\
    X(\omega), 0 \models_n \psi &\Leftrightarrow [X(\omega), 0 \models_n (\lnot p) \wedge (\lnot \Diamond_{(0,8)} p)] \text{ and }[X(\omega), 0 \models_n \phi_3].
\end{align*}
Let $\tau^{(n)}_p(\omega) := \min\{ t\in \N/n ; X(\omega), t \models_n p \}$.
Then $X(\omega), 0 \models (\lnot p) \wedge (\lnot \Diamond_{(0,8)} p)$ and $X(\omega), 0 \models_n (\lnot p) \wedge (\lnot \Diamond_{(0,8)} p)$ means that $\tau_p \geq 8$ and $\tau^{(n)}_p \geq 8$, respectively.
By combining these meanings with the semantics of
$X(\omega), 0 \models \phi_3$
and
$X(\omega), 0 \models_n \phi_3$
respectively, we show that
$X(\omega), 0 \models \psi$ is equivalent to
$\tau_p \in (8,9)$
almost surely, while $X(\omega), 0 \models_n \psi$ is equivalent to $X(\omega), 0 \models_n \bot$ almost surely.

Now we estimate the probability $\Prob(\omega\in\Omega ; X(\omega),0\models \psi)$ of continuous semantics of $\psi$.
 \begin{lem}\label{lemm:isoformula}
    Suppose that $\tau_p(\omega)\geq6$.
    Then $\llbracket\phi_1\rrbracket_\omega$ has an isolated point $\tau_p(\omega)-5$ with positive probability.
    Moreover, $X(\omega),t\not\models\phi_1$ for $t\in[0,\tau_p(\omega)-5)\cup(\tau_p(\omega)-5,\tau_p(\omega)-3)$.
    In other words,
    \begin{align*}
        \begin{cases}
            X(\omega),t\not\models \phi_1 &\text{ for }0\leq t<\tau_p(\omega)-5,\\
            X(\omega),t\models \phi_1 &\text{ for }t=\tau_p(\omega)-5,\\
            X(\omega),t\not\models \phi_1 &\text{ for }\tau_p(\omega)-5<t<\tau_p(\omega)-3.
        \end{cases}
    \end{align*}
\end{lem}

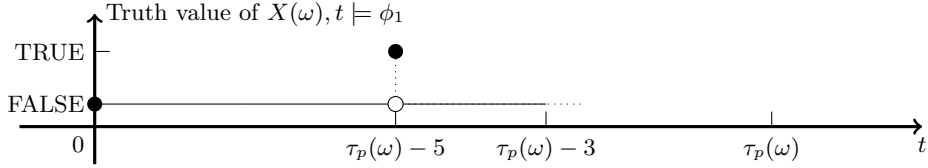
\begin{figure}[ht]
    \centering
    \begin{tikzpicture}
\draw[very thick,->] (-1,0)  --(11,0) node [below]{\footnotesize$t$};
\draw[very thick,->] (0,-0.5)--(0,1.5) node [right]{\footnotesize Truth value of $X(\omega),t\models \phi_1$};
\draw (0,.2)--(0,0) node[anchor=north east]{\footnotesize$0$};
\coordinate (a) at (9,0);
\coordinate (b) at (5,0);
\coordinate (c) at (3,0);
\coordinate (n) at (0,.2);
\coordinate (nh) at (.2,0);
\draw (0,1)node[left]{\footnotesize TRUE}--(.2,1) ;
\draw ($(a)+(n)$)--(a) node[anchor=north]{\footnotesize$\tau_p(\omega)$};
\draw ($(a)-(b)+(n)$)--($(a)-(b)$) node[anchor=north]{\footnotesize$\tau_p(\omega)-5$};
\draw ($(a)-(c)+(n)$)--($(a)-(c)$) node[anchor=north]{\footnotesize$\tau_p(\omega)-3$};
\draw (0,0.3)--(2,.3);
\draw (2,.3)--($(a)-(b)+(0,.3)$);
\draw [dotted] ($(a)-(b)+(0,.3)$)--($(a)-(b)+(0,1)$);
\draw ($(a)-(b)+(0,.3)$)--($(a)-(c)+(0,.3)$);
\draw[dotted] ($(a)-(b)+(0,0.3)$)--(6.5,.3);
\fill (0,0.3) circle[radius=0.1] node[anchor=east]{\footnotesize FALSE};
\filldraw[fill=white,draw=black] ($(a)-(b)+(0,0.3)$) circle[radius=0.1];
\fill ($(a)-(b)+(0,1)$) circle[radius=0.1];
\end{tikzpicture}
    \caption{The truth value of ``$X(\omega),t\models\phi_1$''.}
    \label{fig:phi_1_truth_value}
\end{figure}

\begin{proof}
    Note that $\tau_p(\omega)<\infty$ almost surely.
    Suppose $\tau_p(\omega)\geq6$. 
    Then $X(\omega),t\not\models p$ for $t< \tau_p(\omega)$ and $\inf\{t\geq0;X(\omega),t\models p\}=\tau_p(\omega)$.
    
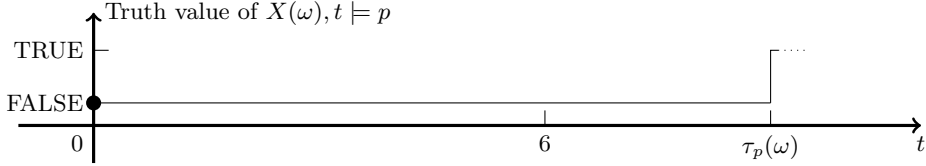
    \begin{tikzpicture}
        \draw[shift={(0,2.5)},very thick,->] (0,-0.5)--(0,1.5) node [right]{\footnotesize Truth value of $X(\omega),t\models p$};
        \draw[shift={(0,2.5)}] (0,.2)--(0,0) node[anchor=north east]{\footnotesize$0$};
        \draw[shift={(0,2.5)},very thick,->] (-1,0)  --(11,0) node [below]{\footnotesize$t$};
        \coordinate[shift={(0,2.5)}] (n) at (0,.2);
        \coordinate[shift={(0,2.5)}] (a) at (9,0);
        \draw[shift={(0,2.5)}] (0,1)node[left]{\footnotesize TRUE}--(.2,1);
        \draw[shift={(0,2.5)}] (6,.2)--(6,0) node [below]{\footnotesize$6$};
        \draw[shift={(0,2.5)}] ($(a)+(n)$)--($(a)$) node [anchor=north] {\footnotesize$\tau_p(\omega)$};
        \draw[shift={(0,2.5)}] (0,.3)--(4,0.3)--($(a)+(0,.3)$);
        \draw[shift={(0,2.5)}] ($(a)+(0,.3)$)--($(a)+(0,1)$)--($(a)+(.1,0)+(0,1)$);
        \draw[shift={(0,2.5)},dotted] ($(a)+(.1,0)+(0,1)$)--(9.5,1);
        \fill[shift={(0,2.5)}] (0,.3)  circle[radius=0.1] node[left]{\footnotesize FALSE};
    \end{tikzpicture}

    From the definition of $\tau_p(\omega)$, if $t\leq\tau_p(\omega)-4$, there is no $s\in(t+1,t+4)$ such that  $X(\omega),s\models p$, which implies $X(\omega),t\not\models\Diamond_{(1,4)}p$.
    Again from the definition of $\tau_p(\omega)$, if $t\in(\tau_p(\omega)-4,\tau_p(\omega)-1)$, there exists some $s\in (t+1,t+4)$ such that $X(\omega),s\models p$.
    Thus we obtain
    
    \begin{align*}
    \begin{cases}
        X(\omega),t\not\models\Diamond_{(1,4)}p&\text{ for }t\leq \tau_p(\omega)-4,\\
        X(\omega),t\models\Diamond_{(1,4)}p&\text{ for }\tau_p(\omega)-4<t<\tau_p(\omega)-1.
    \end{cases}
    \end{align*}
    
    \begin{tikzpicture}
        \draw[shift={(0,3.5)},very thick,->] (0,-0.5)--(0,1.5) node [right]{\footnotesize Truth value of $X(\omega),t\models p$};
        \draw[shift={(0,3.5)}] (0,.2)--(0,0) node[anchor=north east]{$0$};
        \draw[shift={(0,3.5)},very thick,->] (-1,0)  --(11,0) node [below]{$t$};
        \coordinate[shift={(0,3.5)}] (n) at (0,.2);
        \coordinate[shift={(0,3.5)}] (a) at (9,0);
        \draw[shift={(0,3.5)}] (0,1)node[left]{\footnotesize TRUE}--(.2,1);
        \draw[shift={(0,3.5)}] (6,.2)--(6,0) node [below]{\footnotesize$6$};
        \draw[shift={(0,3.5)}] ($(a)+(n)$)--($(a)$) node [anchor=north] {\footnotesize$\tau_p(\omega)$};
        \draw[shift={(0,3.5)}] (0,.3)--(4,0.3)--($(a)+(0,.3)$);
        \draw[shift={(0,3.5)}] ($(a)+(0,.3)$)--($(a)+(0,1)$)--($(a)+(.1,0)+(0,1)$);
        \draw[shift={(0,3.5)},dotted] ($(a)+(.1,0)+(0,1)$)--(9.5,1);
        \fill[shift={(0,3.5)}] (0,.3)  circle[radius=0.1] node[left]{\footnotesize FALSE};

        \coordinate[shift={(0,2.5)}] (a) at (9,0);
        \fill[shift={(0,2.5)},lightgray] ($(a)-(2.5,0)+(1,0)$)--($(a)-(2.5,0)+(4,0)$)--($(a)-(2.5,0)+(4,0)+(0,.2)$)--($(a)-(2.5,0)+(1,0)+(0,.2)$)--cycle;
        \draw[shift={(0,2.5)},very thick,->] ($(a)-(4.5,0)$)--($(a)-(3,0)+(4,0)+(1,0)$);
        \draw[shift={(0,2.5)}] ($(a)-(2.5,0)+(0,.2)$)--($(a)-(2.5,0)$) node[below]{\footnotesize$t$};
        \draw[shift={(0,2.5)}] ($(a)-(2.5,0)+(1,0)$) node[below]{\footnotesize $t+1$};
        \draw[shift={(0,2.5)}] ($(a)-(2.5,0)+(4,0)$) node[below]{\footnotesize $t+4$};
        
        \draw[dotted] (9,3.5)--(9,2.5);
        \draw[dotted] (6.5,2.5)--(6.5,0);
        
        \draw[very thick,->] (0,-0.5)--(0,1.5) node [right]{\footnotesize Truth value of $X(\omega),t\models \Diamond_{(1,4)}p$};
        \draw (0,.2)--(0,0) node[anchor=north east]{$0$};
        \draw (0,1)node[left]{\footnotesize TRUE}--(.2,1);
        \draw[very thick,->] (-1,0)  --(11,0) node [below]{\footnotesize$t$};
        \coordinate (n) at (0,.2);
        \coordinate (a) at (9,0);
        \draw ($(a)+(n)$)--($(a)$) node [anchor=north west] {\footnotesize$\tau_p(\omega)$};
        \draw ($(a)-(4,0)+(n)$)--($(a)-(4,0)$) node[anchor=north]{\footnotesize$\tau_p(\omega)-4$};
        \draw ($(a)-(1,0)+(n)$)--($(a)-(1,0)$) node[anchor=north]{\footnotesize$\tau_p(\omega)-1$};
        \draw (0,.3)--($(a)-(4,0)+(0,.3)$);
        \draw[dotted] ($(a)-(4,0)+(0,.3)$)--($(a)-(4,0)+(0,1)$);
        \draw ($(a)-(4,0)+(0,1)$)--($(a)-(1,0)+(0,1)$);
        \draw[dotted] ($(a)-(1,0)+(0,1)$)--(8.5,1);
        \fill (0,.3)  circle[radius=0.1] node [left]{\footnotesize FALSE};
        \fill ($(a)-(4,0)+(0,.3)$) circle[radius=0.1];
        \filldraw[fill=white,draw=black] ($(a)-(4,0)+(0,1)$) circle[radius=0.1];
    \end{tikzpicture}
    
    Similarly, we can show that 
    \begin{align*}
    \begin{cases}
        X(\omega),t\models\lnot\Diamond_{(1,3)}p&\text{ for }0\leq t\leq\tau_p(\omega)-3,\\
            X(\omega),t\not\models\lnot\Diamond_{(1,3)}p&\text{ for }\tau_p(\omega)-3<t<\tau_p(\omega)-1.
    \end{cases}
    \end{align*}
    
    \begin{tikzpicture}
        \draw[very thick,->] (0,-0.5)--(0,1.5) node [right]{\footnotesize Truth value of $X(\omega),t\models\lnot \Diamond_{(1,3)}p$};
        \draw (0,.3)node[left]{\footnotesize FALSE}--(.2,.3);
        \draw (0,.2)--(0,0) node[anchor=north east]{\footnotesize$0$};
        \draw[very thick,->] (-1,0)  --(11,0) node [below]{\footnotesize$t$};
        \coordinate (n) at (0,.2);
        \coordinate (a) at (9,0);
        \draw ($(a)+(n)$)--($(a)$) node [anchor=north west] {\footnotesize$\tau_p$};
        \draw ($(a)-(3,0)+(n)$)--($(a)-(3,0)$) node[anchor=north]{\footnotesize$\tau_p-3$};
        \draw ($(a)-(1,0)+(n)$)--($(a)-(1,0)$) node[anchor=north]{\footnotesize$\tau_p-1$};
        \draw (0,1)--($(a)-(3,0)+(0,1)$);
        \draw[dotted] ($(a)-(3,0)+(0,1)$)--($(a)-(3,0)+(0,.3)$);
        \draw ($(a)-(3,0)+(0,.3)$)--($(a)-(1.5,0)+(0,.3)$)--($(a)-(1,0)+(0,.3)$);
        \draw[dotted] ($(a)-(1,0)+(0,.3)$)--(8.5,.3);
        \fill (0,1) circle[radius=0.1] node [left]{\footnotesize TRUE};
        \fill ($(a)-(3,0)+(0,1)$) circle[radius=0.1];
        \filldraw[fill=white,draw=black] ($(a)-(3,0)+(0,.3)$) circle[radius=0.1];
    \end{tikzpicture}
    
    Consequently, $X(\omega),t\models\Diamond_{(1,4)}p\wedge(\lnot\Diamond_{(1,3)}p)$ does not hold for any $t\in[0,\tau_p(\omega)-4]$ and $t\in(\tau_p(\omega)-3,\tau_p(\omega)-1)$, but it holds for $(\tau_p(\omega)-4,\tau_p(\omega)-3]$.
    Namely,
    \begin{align*}
    \begin{cases}
        X(\omega),t\not\models\Diamond_{(1,4)}p\wedge(\lnot\Diamond_{(1,3)}p)&\text{ for }0\leq t\leq \tau_p(\omega)-4,\\
        X(\omega),t\models \Diamond_{(1,4)}p\wedge(\lnot\Diamond_{(1,3)}p)&\text{ for }\tau_p-4(\omega)<t\leq\tau_p(\omega)-3,\\
        X(\omega),t\not\models \Diamond_{(1,4)}p\wedge(\lnot\Diamond_{(1,3)}p)&\text{ for }\tau_p(\omega)-3<t<\tau_p(\omega)-1.
    \end{cases}
    \end{align*}

    \begin{tikzpicture}
        \draw[very thick,->] (0,-0.5)--(0,1.5) node [right]{\footnotesize Truth value of 
 $X(\omega),t\models \Diamond_{(1,4)}p\wedge(\lnot\Diamond_{(1,3)}p)$};
        \draw (0,1)node[left]{\footnotesize TRUE}--(.2,1);
        \draw[very thick,->] (-1,0)  --(11,0) node [below]{\footnotesize$t$};
        \draw (0,.2)--(0,0) node[anchor=north east]{\footnotesize$0$};
        \coordinate (n) at (0,.2);
        \coordinate (a) at (9,0);
        \draw ($(a)+(0,.2)$)--($(a)$) node [anchor=north]{\footnotesize$\tau_p$};
        \foreach \x in {1,3,4}
        {
        \draw ($(a)-(\x,0)+(0,.2)$)--($(a)-(\x,0)$) node [anchor=north]{\tiny$\tau_p-\x$};
        }
        \draw (0,.3)--($(a)-(4,0)+(0,.3)$) ;
        \draw[dotted] ($(a)-(4,0)+(0,.3)$)--($(a)-(4,0)+(0,1)$);
        \draw ($(a)-(4,0)+(0,1)$)--($(a)-(3,0)+(0,1)$);
        \draw[dotted] ($(a)-(3,0)+(0,1)$)--($(a)-(3,0)+(0,.3)$);
        \draw ($(a)-(3,0)+(0,.3)$)--($(a)-(1,0)+(0.1,0)+(0,.3)$);
        \draw[dotted] ($(a)-(1,0)+(0,.3)$)--(8.5,.3);
        \fill (0,.3) circle[radius=0.1] node[left]{\footnotesize FALSE};
        \fill ($(a)-(4,0)+(0,.3)$) circle[radius=0.1];
        \filldraw[fill=white,draw=black] ($(a)-(4,0)+(0,1)$) circle[radius=0.1];
        \fill ($(a)-(3,0)+(0,1)$) circle[radius=0.1];
        \filldraw[fill=white,draw=black] ($(a)-(3,0)+(0,.3)$) circle[radius=0.1];
    \end{tikzpicture}
    
    To satisfy $X(\omega),t\models \phi_1$, it must holds that
    $X(\omega),s\models \Diamond_{(1,4)}p\wedge(\lnot\Diamond_{(1,3)}p)$
    for every $s\in (t+1,t+2)$.
    Then $X(\omega),t\models\phi_1$ does not holds for $t\in[0,\tau_p(\omega)-5)$ and $t\in(\tau_p(\omega)-5,\tau_p(\omega)-3)$ and holds on $t=\tau_p(\omega)-5$.
    We obtain the required claim since such an isolated point occurs whenever $\tau_p(\omega)\geq6$.

   
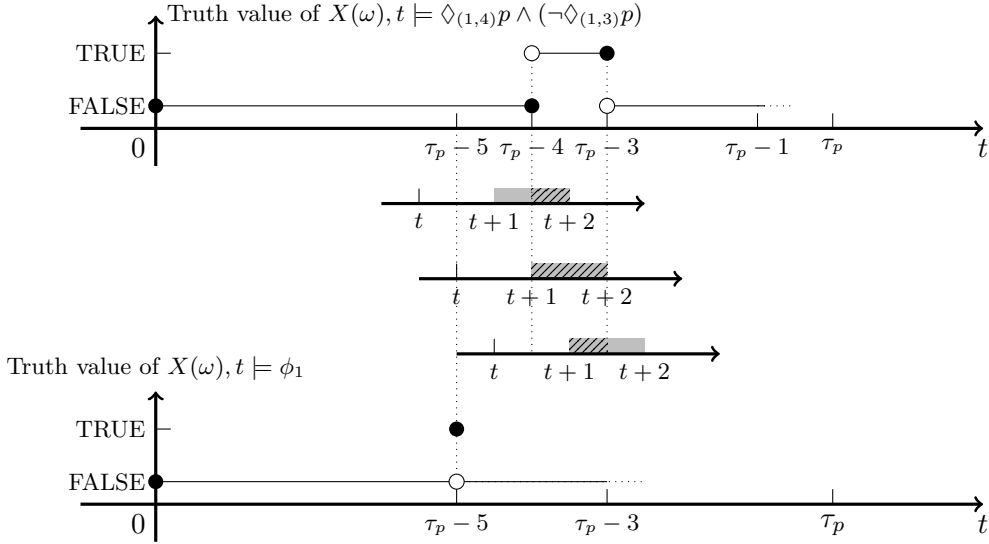
\begin{tikzpicture}
        \draw[very thick,->] (0,-0.5)--(0,1.5) node [right]{\footnotesize Truth value of $X(\omega),t\models \Diamond_{(1,4)}p\wedge(\lnot\Diamond_{(1,3)}p)$};
        \draw (0,1)node[left]{\footnotesize TRUE}--(.2,1);
        \draw[very thick,->] (-1,0)  --(11,0) node [below]{$t$};
        \draw (0,.2)--(0,0) node[anchor=north east]{$0$};
        \coordinate (n) at (0,.2);
        \coordinate (a) at (9,0);
        \draw ($(a)+(0,.2)$)--($(a)$) node [anchor=north]{\footnotesize$\tau_p$};
        \foreach \x in {1,3,4,5}
        {
        \draw ($(a)-(\x,0)+(0,.2)$)--($(a)-(\x,0)$) node [anchor=north]{\footnotesize$\tau_p-\x$};
        }
        \draw (0,.3)--($(a)-(4,0)+(0,.3)$) ;
        \draw[dotted] ($(a)-(4,0)+(0,.3)$)--($(a)-(4,0)+(0,1)$);
        \draw ($(a)-(4,0)+(0,1)$)--($(a)-(3,0)+(0,1)$);
        \draw[dotted] ($(a)-(3,0)+(0,1)$)--($(a)-(3,0)+(0,.3)$);
        \draw ($(a)-(3,0)+(0,.3)$)--($(a)-(1,0)+(0.1,0)+(0,.3)$);
        \draw[dotted] ($(a)-(1,0)+(0,.3)$)--(8.5,.3);
        \fill (0,.3) circle[radius=0.1] node[left]{\footnotesize FALSE};
        \fill ($(a)-(4,0)+(0,.3)$) circle[radius=0.1];
        \filldraw[fill=white,draw=black] ($(a)-(4,0)+(0,1)$) circle[radius=0.1];
        \fill ($(a)-(3,0)+(0,1)$) circle[radius=0.1];
        \filldraw[fill=white,draw=black] ($(a)-(3,0)+(0,.3)$) circle[radius=0.1];
        \coordinate (d) at (0.5,-1);
        \coordinate (o) at ($(a)-(6,0)$);
        \foreach \i in {1,2,3}
        {
        \fill[lightgray] ($(o)+\i*(d)+(0,.2)+(1,0)$)--($(o)+\i*(d)+(1,0)$)--($(o)+\i*(d)+(2,0)$)--($(o)+\i*(d)+(0,.2)+(2,0)$)--cycle;
        \draw[very thick,->]($(o)-(0.5,0)+\i*(d)$)--($(o)+(3,0)+\i*(d)$);
        \draw ($(o)+\i*(d)+(0,.2)$)--($(o)+\i*(d)$) node[below]{\footnotesize$t$};
        \draw ($(o)+\i*(d)+(1,0)$) node[below]{\footnotesize$t+1$};
        \draw ($(o)+\i*(d)+(2,0)$) node[below]{\footnotesize$t+2$};
        \begin{scope}\clip ($(a)-(5,0)+(1,0)$)--($(a)-(5,0)+(1,0)+(0,-3)$)--($(a)-(5,0)+(2,0)+(0,-3)$)--($(a)-(5,0)+(2,0)$)--cycle;
            \fill[pattern=north east lines] ($(o)+\i*(d)+(0,.2)+(1,0)$)--($(o)+\i*(d)+(1,0)$)--($(o)+\i*(d)+(2,0)$)--($(o)+\i*(d)+(0,.2)+(2,0)$)--cycle;
        \end{scope}
        }
        \foreach \i in {1,2}
        {
        \draw[dotted] ($(a)-(5,0)+\i*(1,0)$)--($(a)-(5,0)+\i*(1,0)+(0,-3)$);
        }
        \draw[dotted] ($(a)-(5,0)$)--($(a)-(5,0)+(0,-4)$);

        \draw[shift={(0,-5)},very thick,->] (-1,0)  --(11,0) node [below]{$t$};
        \draw[shift={(0,-5)},very thick,->] (0,-0.5)--(0,1.5) node [above]{\footnotesize Truth value of $X(\omega),t\models \phi_1$};
        \draw[shift={(0,-5)}] (0,.2)--(0,0) node[anchor=north east]{$0$};
        \coordinate[shift={(0,-5)}] (a) at (9,0);
        \coordinate[shift={(0,-5)}] (b) at (5,0);
        \coordinate[shift={(0,-5)}] (c) at (3,0);
        \coordinate[shift={(0,-5)}] (n) at (0,.2);
        \coordinate[shift={(0,-5)}] (nh) at (.2,0);
        \draw[shift={(0,-5)}] (0,1)node[left]{\footnotesize TRUE}--(.2,1) ;
        \draw[shift={(0,-5)}] ($(a)+(n)$)--(a) node[anchor=north]{$\tau_p$};
        \draw[shift={(0,-5)}] ($(a)-(b)+(n)$)--($(a)-(b)$) node[anchor=north]{\footnotesize$\tau_p-5$};
        \draw[shift={(0,-5)}] ($(a)-(c)+(n)$)--($(a)-(c)$) node[anchor=north]{\footnotesize$\tau_p-3$};
        \draw[shift={(0,-5)}] (0,0.3)--($(a)-(b)+(0,.3)$);
        \draw[shift={(0,-5)}] [dotted] ($(a)-(b)+(0,.3)$)--($(a)-(b)+(0,1)$);
        \draw[shift={(0,-5)}] ($(a)-(b)+(0,.3)$)--($(a)-(c)+(0,.3)$);
        \draw[shift={(0,-5)},dotted] ($(a)-(b)+(0,0.3)$)--(6.5,.3);
        \fill[shift={(0,-5)}] (0,0.3) circle[radius=0.1] node[anchor=east]{\footnotesize FALSE};
        \filldraw[shift={(0,-5)},fill=white,draw=black] ($(a)-(b)+(0,0.3)$) circle[radius=0.1];
        \fill[shift={(0,-5)}] ($(a)-(b)+(0,1)$) circle[radius=0.1];     
    \end{tikzpicture}
\end{proof}

\begin{lem}\label{lemm:count_exam_conti_sem}
 $X(\omega),0\models\psi$ is equivalent to $\tau_p(\omega)\in(8,9)$ almost surely.
 In particular, $\Prob(X(\omega),0\models\psi)>0$.
\end{lem} 

\begin{proof}
    First, we note that  $X(\omega),0\models(\lnot p)\wedge(\lnot\Diamond_{(0,8)}p)$ is equivalent to $\tau_p(\omega)\geq8$.
    Since $\tau_p(\omega)$ has density, $X(\omega),0\models\psi$ implies $\tau_p(\omega)>8$ almost surely.
    
    Suppose that $\tau_p(\omega)>8$.
    Define
    \begin{align*}
        \tau_1(\omega):=\inf\{t; X(\omega),t\models\phi_1\},\\
        \tau_2(\omega):=\inf\{t; X(\omega),t\models\phi_2\}.
    \end{align*}
    Then, from Lemma~\ref{lemm:isoformula},
    \begin{align}
    \begin{cases}
        X(\omega),t\not\models\phi_1,&\text{ for }t\in[0,\tau_p(\omega)-5),\\
        X(\omega),t\models\phi_1,&\text{ at }t=\tau_p(\omega)-5,\\
        X(\omega),t\not\models\phi_1,&\text{ for }t\in(\tau_p(\omega)-5,\tau_p(\omega)-3).
    \end{cases}\label{eq:tau_1_truthvalue}
    \end{align}
    Hence $\tau_1(\omega)=\tau_p(\omega)-5$.
    Furthermore, $X(\omega),t\models\phi_2$ means
        \begin{align*}
            \begin{cases}
                X(\omega),s\not\models\phi_1,&\text{ for }s\in(t+1,t+2),\\
                X(\omega),s\models\phi_1,&\text{ at }s=t+2,\\
                 X(\omega),s\not\models\phi_1,&\text{ for }s\in(t+2,t+3).
            \end{cases}
        \end{align*}
        Then we can conclude from \eqref{eq:tau_1_truthvalue} that
        \begin{align*}
        \begin{cases}
            X(\omega),t\models\lnot\phi_2,&\text{ for }t\in[0,\tau_p(\omega)-7)\\
            X(\omega),t\models\phi_2,&\text{ at }t=\tau_p(\omega)-7,\\
            X(\omega),t\models\lnot\phi_2,&\text{ for }t\in(\tau_p(\omega)-7,\tau_p(\omega)-5).
        \end{cases}
        \end{align*}
        This means exactly $\tau_2(\omega)=\tau_p(\omega)-7$.
        
        Suppose that $\tau_p(\omega)\geq9$.
        Then $\tau_2(\omega)=\tau_p(\omega)-7\geq2$.
        Since $X(\omega),0\models\phi_3$ is equivalent to $\tau_2(\omega)\in(1,2 )$, $X(\omega),0\models\phi_3$ does not hold.
        Then $X(\omega),0\models\psi$ implies $\tau_p(\omega)<9$.
        Consequently, we obtain that $X(\omega),0\models\psi$ implies $\tau_p(\omega)\in(8,9)$ almost surely.
        
        Conversely, as we have seen that $\tau_p(\omega)>8$ implies $\tau_2(\omega)=\tau_p(\omega)-7$, $\tau_p(\omega)\in(8,9)$ implies $\tau_2(\omega)\in(1,2) $.
        Then $X(\omega),0\models\phi_3$, and together with $\tau_p>8$, we can conclude that $X(\omega),0\models\psi$.        
\end{proof}

\begin{lem}\label{lemm:equidisc}
    Let $n\geq2$, $\tau_p^{(n)}(\omega):=\inf\{t\in\N/n; X(\omega),t\models_n p\}$, and suppose that $\tau^{(n)}_p(\omega)\geq6$.
    Then it holds that
    \begin{align*}
    \begin{cases}
        X(\omega),t\not\models_n\phi_1&\text{ for }t=0,1/n,\cdots,\tau_p^{(n)}(\omega)-5-1/n,\\
	X(\omega),t\models_n\phi_1&\text{ for }t=\tau_p^{(n)}(\omega)-5,\tau_p^{(n)}(\omega)-5+1/n,\\
	X(\omega),t\not\models_n\phi_1&\text{ for }t=\tau_p^{(n)}(\omega)-5+2/n,\cdots,\tau_p^{(n)}(\omega)-2-2/n.
    \end{cases}
    \end{align*}
\end{lem}
	
\begin{proof}
    By the definition of diamond operator, $X(\omega),t\models_n\Diamond_{(1,4)}p$ is equivalent to  
    \begin{align*}
        (\exists s\in\{t+1+1/n,t+1+2/n\cdots t+4-1/n\})[X,s\models_np].
    \end{align*}
    Then we observe from the definition  of $\tau_p^{(n)}(\omega)$ that
    \begin{align*}
        \begin{cases}
            X(\omega),t\not\models_n\Diamond_{(1,4)}p&\text{ for }t=0,1/n,\cdots,\tau_p^{(n)}(\omega)-4\\
            X(\omega),t\models_n\Diamond_{(1,4)}p&\text{ for }t=\tau_p^{(n)}(\omega)-4+1/n,\tau_p^{(n)}(\omega)-4+2/n, \cdots,\tau_p^{(n)}(\omega)-1-1/n.
        \end{cases}
    \end{align*}
    Similarly, we have 
    \begin{align*}
        \begin{cases}
            X(\omega),t\models_n \Diamond_{(1,3)}p&\text{ for }t=0,1/n,\cdots,\tau_p^{(n)}(\omega)-3\\
            X(\omega),t\not\models_n \Diamond_{(1,3)}p&\text{ for }t=\tau_p^{(n)}(\omega)-3+1/n,\tau_p^{(n)}(\omega)-3+2/n, \cdots,\tau_p^{(n)}(\omega)-1-1/n.
        \end{cases}
    \end{align*}
    Then we obtain
    \begin{align*}
        \begin{cases}
            X(\omega),t\not\models_n\Diamond_{(1,4)}p\wedge\lnot\Diamond_{(1,3)}p&\text{ for }t=0,\cdots,\tau_p^{(n)}(\omega)-4,\\
            X(\omega),t\models_n\Diamond_{(1,4)}p\wedge\lnot\Diamond_{(1,3)}p&\text{ for }t=\tau_p^{(n)}(\omega)-4+1/n,\cdots\tau_p^{(n)}(\omega)-3,\\
            X(\omega),t\not\models_n\Diamond_{(1,4)}p\wedge\lnot\Diamond_{(1,3)}p&\text{ for }t=\tau_p^{(n)}(\omega)-3+1/n,\cdots,\tau_p^{(n)}(\omega)-1-1/n.
        \end{cases}
    \end{align*}
    From the definition of Box operator, $X,t\models_n\phi_1$ is equivalent to 
    \begin{align*}
        (\forall s\in\{t+1+1/n,\cdots,t+2-1/n\})[X(\omega),t\models_n\Diamond_{(1,4)}p\wedge\lnot\Diamond_{(1,3)}p].
    \end{align*}
    Then we observe that
        \begin{align*}
    \begin{cases}
        X(\omega),t\not\models_n\phi_1&\text{ for }t=0,1/n,\cdots,\tau_p^{(n)}(\omega)-5-1/n,\\
	X(\omega),t\models_n\phi_1&\text{ for }t=\tau_p^{(n)}(\omega)-5,\tau_p^{(n)}(\omega)-5+1/n,\\
	X(\omega),t\not\models_n\phi_1&\text{ for }t=\tau_p^{(n)}(\omega)-5+2/n,\cdots,\tau_p^{(n)}(\omega)-2-2/n.
    \end{cases}\tag*{\begin{minipage}{25pt}
        \vphantom{$\tau_p^{(n)}$} \\
        \vphantom{$\tau_p^{(n)}$} \\
        \vphantom{$\tau_p^{(n)}$} \qedhere
    \end{minipage}}
    \end{align*}
\end{proof}

\begin{lem}\label{lemm:count_exam_disc_sem}
    Let $n\geq2$.
    Then $X(\omega),0\not\models_n\psi$ for every $\omega\in\Omega$.
\end{lem}
\begin{proof}
    Define $\tau_p^{(n)}(\omega)$ as Lemma~\ref{lemm:equidisc}.
    Since $X(\omega),0\models_n\psi$ implies $X(\omega),0\models_n(\lnot p)\wedge(\lnot\Diamond_{(0,8)}p)$,
    $\tau_p^{(n)}(\omega)\geq8$.
    Then we obtain from Lemma~\ref{lemm:equidisc} that
    \begin{align*}
    \begin{cases}
        X(\omega),t\not\models_n\phi_1&\text{ for }t=0,1/n,\cdots,\tau_p^{(n)}(\omega)-5-1/n,\\
	X(\omega),t\models_n\phi_1&\text{ for }t=\tau_p^{(n)}(\omega)-5,\tau_p^{(n)}(\omega)-5+1/n,\\
		X(\omega),t\not\models_n\phi_1&\text{ for }t=\tau_p^{(n)}(\omega)-5+2/n,\cdots,\tau_p^{(n)}(\omega)-2-2/n.
    \end{cases}
    \end{align*}
    From the definition of the discrete semantics, $X(\omega),t\models_n\phi_2$ is equivalent to 
    \begin{align*}
    \begin{cases}
        X(\omega),s\not\models_n\phi_1&\text{ for }s=t+1+1/n,\cdots,t+2-1/n,\\
	X(\omega),s\models_n\phi_1&\text{ for }s=t+2,\\
	X(\omega),s\not\models_n\phi_1&\text{ for }s=t+2+1/n,\cdots,t+3-1/n.
    \end{cases}
    \end{align*}
    In other words, for $X(\omega),t\models_n\phi_2$ to hold, $X(\omega),s\models_n\phi_1$ must hold exactly on $s=t+2$, and $X(\omega),s\models_n\phi_1$ must not hold for other $s\in\N$ such that $t+1+1/n\leq s\leq t+3-1/n$.
    However, $X(\omega),s\models_n\phi_1$ holds at two adjacent $s$ values, namely $s=\tau_p^{(n)}(\omega)-5$ and $s=\tau_p^{(n)}(\omega)-5+1/n$, and does not hold for other $s$ values such that $0\leq s\leq \tau^{(n)}_p(\omega)-2-2/n$.
    Therefore, $X(\omega),t\models_n\phi_2$ does not hold as long as $t+2\leq\tau^{(n)}_p(\omega)-2-2/n$, which implies $t\leq\tau^{(n)}_p(\omega)-4-2/n$.
    Since $\tau_p(\omega)\geq8$, $X(\omega),t\models_n\phi_2$ does not hold as long as $t\leq4-2/n$ and hence $X(\omega),0\not\models_n\Diamond_{(1,2)}\phi_2$.
\end{proof}

\begin{thm}
    Put $\psi$ as Lemma~\ref{lemm:count_exam_conti_sem}.
    Then $\Prob(\omega; X(\omega),0\models_n\phi_3)$ does not converges to $\Prob(\omega; X(\omega),0\models\phi_3)$.
\end{thm}

\begin{proof}
    From Lemma~\ref{lemm:count_exam_conti_sem}, we have $\Prob(\omega; X(\omega),0\models\psi)=\Prob(\omega; \tau_p(\omega)\in(8,9))>0$.
    On the other hand, from Lemma~\ref{lemm:count_exam_disc_sem}, we have $\Prob(\omega;X(\omega),0\models_n\psi)=0$ for every $n$ larger than $2$.
    Therefore, $\Prob(\omega;X(\omega),0\models_n\psi)$ never converges to $\Prob(\omega;X(\omega),0\models\psi)$.
\end{proof}

\begin{rem}
    In the above counterexample, the subscripted intervals for the diamond operator are restricted to be open. If we restrict the intervals to be left-- or right--open, we can construct similar counterexamples. However, it remains an open problem whether we can create a counterexample when we restrict all intervals to be closed.
\end{rem}

\section{Discretization of MTL--formula: Convergence Result of\texorpdfstring{\\}{} MTL--formula without until and nested temporal operators}\label{sec:flat_MTL}

In the previous section, we presented a counterexample of an MTL--formula, illustrating a case where its probability in discrete semantics does not converge to that in continuous semantics.
In contrast, in this section, we establish the convergence of such probabilities by restricting MTL--formulas in which no until operator and nested diamond and box operators appear.
We refer to these restricted formulas as $\flat$MTL--formulas.
While we discussed MTL--formulas for a Brownian motion in the preceding sections, we will now present the convergence result for general one-dimensional stochastic differential equations:

\begin{align}
    \begin{cases}
dX_t=b(X_t)dt+\sigma(X_t)dW_t,\\
X_0=\xi\in\R.
\end{cases}\label{eq:SDE_with_drift}
\end{align}

We impose the following conditions on the SDE to establish the convergence result for $\flat$MTL--formulas \eqref{eq:SDE_with_drift}.
These conditions are also sufficient to ensure the solution's existence, uniqueness, and absolute continuity.(see Appendix~\ref{append:SDE}):
\begin{asm}\label{cond:SDE_with_drift}\hfill
	\begin{enumerate}[(i)]
		\item For every compact set $ K\subset\R$, $\inf\sigma(K)>0$.
		\item $\sigma$ is Lipschitz continuous.
		\item $b$ is bounded and Borel measurable.
	\end{enumerate}
\end{asm}

Now let us define $\flat$MTL--formulas rigorously.
\begin{defi}[Syntax of $\flat$MTL--formula]\label{defi:syntax_flat_MTL}
	Let $AP$ be a finite set of atomic formulas.
	We define the syntax of $\flat$MTL by the following induction.
	\begin{enumerate}[(i)]\label{defi:syntax_flat_MTL_general}
		\item All atomic formulas are $\flat$MTL--formulas.
		\item If $\phi$ is an $\flat$MTL--formula, $\lnot \phi$ is an $\flat$MTL--formula.
		\item If $\phi_1$ and $\phi_2$ are $\flat$MTL--formulas, then $\phi_1\wedge\phi_2$ is an $\flat$MTL--formula.
		\item If $p$ is a propositional formula and $I$ is a positive interval on $[0,\infty)$, then $\Diamond_I p$ is an $\flat$MTL--formula.  We assume that $I = \langle S, T \rangle$ where $0 \leq S < T$ in which $T$ can be infinite $\infty$ and the boundaries of $I$ can be either open or closed.
	\end{enumerate}
Here, we define the propositional formula as follows:
\begin{enumerate}[(a)]
     \item All atomic formulas are propositional formulas.
    \item If $p$ is a propositional formula, $\lnot p$ is a propositional formula.
    \item If $p_1$ and $p_2$ are propositional formulas, then $p_1\wedge p_2$ is a propositional formula.
\end{enumerate}
\end{defi}

\begin{rem}
    For a MTL--formula $\phi$, $\square_I \phi$ is represented by $\neg \diamond_I \neg \phi$ using De Morgan duality.
    We focus $\square$ and $\diamond$ operators and does not treat the until operator.
    The condition of convergence for the until operator is a future work.
\end{rem}

The semantics of $\flat$MTL are given in the same way as MTL--formulas.
\begin{defi}[Semantics of $\flat$MTL--formulas]\label{defi:semantics_flat_MTL}
Let $B_i,\ i=1,\cdots,k$ be Borel sets on $\R$ and $AP=\{a_i ; i=1,\cdots,k\}$ be the set of $k$ atomic formulas.
The semantics of $\flat$MTL--formulas are defined inductively as follows.
    \begin{enumerate}[(i)]
        \item $X(\omega),t\models a_i \Leftrightarrow X_t(\omega)\in B_i$ for $i=1,\cdots,k$.
        \item $X(\omega),t\models \phi_1\wedge \phi_2$ is equivalent to $X(\omega),t\models \phi_1\text{ and } X(\omega),t\models\phi_2$.
        \item $X(\omega),t\models \Diamond_{I}p$ is equivalent to $(\exists s\in I)[X(\omega),t+s\models p]$.
    \end{enumerate}
\end{defi}

We will show the convergence result for $\flat$MTL in Section~\ref{sec:general_flat_convergence}. 
We show the convergence of the probability by showing the convergence of the indicator function of $\flat$MTL--formulas.
Let us define the indicator functions for MTL--formulas as follows:
\begin{defi}\label{defi:indicator_func}
    Let $\phi$ be an $\flat$MTL--formula and define random indicator functions  $\chi_{\phi}(\omega,t)$ and $\chi^{(n)}_{\phi}(t)$ as
	\begin{align*}
		\chi_{\phi}(\omega,t)&:=\begin{cases}
		1&\text{ if }X(\omega),t\models\phi\\
		0&\text{ if }X(\omega),t\not\models\phi,
		\end{cases}\\
		\chi^{(n)}_{\phi}(\omega,t)&:=\begin{cases}
		1&\text{ if }X(\omega),\Lambda_n(t)\models_n\phi\\
		0&\text{ if }X(\omega),\Lambda_n(t)\not\models_n\phi,
		\end{cases}
	\end{align*}
	where $\Lambda_n(t):=\frac{\lfloor nt\rfloor}{n}$.
\end{defi}

The convergence of the indicator function for a formula implies the convergence of the probability of the formula.
More precisely, our proof of the convergence is based on the following lemma:
\begin{lem}\label{lemm:prob_convergence_from_indicator}
	Suppose that $\chi^{(n)}_\phi(\omega,t)\rightarrow\chi_\phi(\omega,t)$ almost surely.
	Then $\Prob(\omega; X(\omega),\Lambda_n(t)\models_n\phi)\rightarrow\Prob(\omega; X(\omega),t\models\phi)$ as $n\rightarrow\infty$.
\end{lem}
\begin{proof}
    From the definition of $\chi_\phi(\omega,t)$ and $\chi_\phi^{(n)}(\omega,t)$, $\chi_\phi(\omega,t)=1$ and $\chi^{(n)}_\phi(\omega,t)=1$ is equivalent to $X(\omega),t\models\phi$ and $X(\omega),\Lambda_n(t)\models_n\phi$, respectively.
    Then $\Prob(\omega\in\Omega; X(\omega),t\models\phi)=\E[\chi_\phi(\omega,t)]$ and $\Prob(\omega\in\Omega ;X(\omega),\Lambda_n(t)\models_n\phi)=\E[\chi_\phi^{(n)}(\omega,t)]$.
    Since $\chi_\phi(\omega,t)\leq1$, 
    $\chi_\phi^{(n)}(\omega,t)\leq1$, and 
    $\E[1]=1$, we can apply \emph{Lebesgue's dominated convergence theorem} (see Theorem 1.34 in \cite{ouc.200337561119660101}) to observe 
    $\E[\chi_\phi^{(n)}(\omega,t)]\rightarrow\E[\chi_\phi(\omega,t)]$.
\end{proof}

\subsection{\texorpdfstring{The case of $\Diamond_{\lrangle{S,T}}p$ with $p$ corresponding to a union of intervals}{The case of a diamond operator on an atomic p corresponding to a union of intervals}}

Before proving the convergence of general $\flat$MTL--formulas, we first show the convergence for a particular type of $\flat$MTL--formulas.
In the subsequent subsection, we will present the proof for the convergence of $\flat$MTL--formulas in the general case.
In this subsection, we consider $\flat$MTL--formulas of the form $\Diamond_{\lrangle{S, T}}p$, where $p$ is a propositional formula.
Here, $\lrangle{S,T}$ denotes a positive interval on $[0,\infty)$, specifically, $\lrangle{S,T}$ represents an interval with endpoints $S$ and $T$ such that $0\leq S<T$.
Note that the interval $\lrangle{S,T}$ can be open, left open, right open, or closed.
In the proof of convergence in this case, we utilize deep insights from stochastic calculus, namely, the notion of \emph{local maxima and local minima} of SDE (see Definition~\ref{defi:local_extreme}), and \emph{the dense property of the zero set} of SDE (see Lemma~\ref{lemm:local_extreme}).

To prove the convergence in the particular case of $\Diamond_{\lrangle{S,T}}p$, we will introduce certain conditions on the propositional formula $p$.

\begin{defi}\label{cond:atomicdisjoint}
Let $B_1,\cdots,B_n$ be a finite family of Borel sets on $\R$.
A pair $B_i,B_j$ is said to be \emph{separated} when $\overline{B_i}\cap\overline{B_j}=\emptyset$.
We say the set $\{B_1,\cdots,B_n\}$ is \emph{pairwise separated} when all pairs of different elements are separated.
\end{defi}

Now we prove the following theorem:

\begin{thm}\label{theo:pairwise_disjoint_diamond}
Let $X$ be the strong solution of \eqref{eq:SDE_with_drift} satisfying Assumption~\ref{cond:SDE_with_drift}.
Let $p$ be an MTL--formula such that $X(\omega),t\models p$ is equivalent to $X_t\in B_p$, where 
\begin{align}
    B_p:=\bigcup_{i=1}^k\lrangle{x_i,y_i}\label{eq:union_of_intervals}
\end{align}
is a union of pairwise separated positive intervals $\{\lrangle{x_i,y_i}; i=1,\cdots,k\}$ on $\R$.
Here $B_p$ possibly equals the empty set or $\R$. 
Define $X(\omega),t\models_n p$ similarly.
Define $\phi:=\Diamond_{\langle S,T\rangle}p$, and $\psi:=\Box_{\lrangle{S,T}}p$ where $\langle S,T\rangle$ is a positive interval on $[0,\infty)$.
Then the following statements hold:
\begin{align*}
    \chi_\phi^{(n)}(\omega)&\overset{n\rightarrow\infty}{\longrightarrow}\chi_\phi(\omega),\asure,\\
    \chi_\psi^{(n)}(\omega)&\overset{n\rightarrow\infty}{\longrightarrow}\chi_\psi(\omega),\asure
\end{align*}
In particular,
\begin{align*}
    \Prob(\omega; X(\omega),\Lambda_n(t)\models_n\phi)\overset{n\rightarrow\infty}{\rightarrow}\Prob(\omega; X(\omega),t\models\phi),\\
    \Prob(\omega; X(\omega),\Lambda_n(t)\models_n\psi)\overset{n\rightarrow\infty}{\rightarrow}\Prob(\omega; X(\omega),t\models\psi).
\end{align*}
\end{thm}

The key to proving the convergence of MTL--formulas with the diamond operator lies in the following inclusions:
\begin{align}
	\bracketo{\phi}&\subset\overline{\internal\bracketo{\phi}}\ \text{almost surely},\label{eq:phi_closure_of_int}\\
	\bracketo{\lnot\phi}&\subset\overline{\internal\bracketo{\lnot\phi}}\ \text{almost surely},\label{eq:notphi_closure_of_int}
\end{align}
where the time set $\bracketo{\phi}$ of MTL--formula $\phi$ is defined in Definition~\ref{defi:time_set}.

To show Theorem~\ref{theo:pairwise_disjoint_diamond}, we will first prove a simplified version of the theorem in which the propositional formula corresponds to an interval.
First, let us show \eqref{eq:phi_closure_of_int} and \eqref{eq:notphi_closure_of_int} in this case:

\begin{lem}\label{lemm:interval_propositional_inter}
Let $X=\{X_t\}_{t\geq0}$ be the strong solution of the SDE~\eqref{eq:SDE_with_drift} satisfying Assumption~\ref{cond:SDE_with_drift}.
Let $p$ be a propositional formula defined as $X(\omega),t\models p\Leftrightarrow X_t(\omega)\in \lrangle{y_1,y_2}$, where $\lrangle{y_1,y_2}$ is a positive interval on $\R$.
Then $p$ satisfies \eqref{eq:phi_closure_of_int} and \eqref{eq:notphi_closure_of_int} almost surely.
Namely, 
\begin{align*}
    \bracketo{p}&\subset\overline{\internal\bracketo{p}},\asure,\\
    \bracketo{\lnot p}&\subset\overline{\internal\bracketo{\lnot p}},\asure
\end{align*}
\end{lem}

To prove this lemma, we have to introduce {\it local minima} and {\it local maxima} of $X$:
\begin{defi}\label{defi:local_extreme}\hfill
\begin{enumerate}[(i)]
    \item Let $f:[0,\infty)\rightarrow\R$ be a given function.
		A number $t\geq0$ is called a {\it point of local maximum}, if there exists a number $\delta>0$ with $f(s)\leq f(t)$ valid for every $s\in[(t-\delta)^+,t+\delta]$; and a {\it point of strict local maximum}, if there exists a number $\delta>0$ with $f(s)<f(t)$ valid for  every $s\in[(t-\delta)^+,t+\delta]\setminus\{t\}$.
    \item Let $f:[0,\infty)\rightarrow\R$ be a given function.
		A number $t\geq0$ is called a {\it point of local minimum}, if there exists a number $\delta>0$ with $f(s)\geq f(t)$ valid for every $s\in[(t-\delta)^+,t+\delta]$; and a {\it point of strict local minimum}, if there exists a number $\delta>0$ with $f(s)>f(t)$ valid for every $s\in[(t-\delta)^+,t+\delta]\setminus\{t\}$.
\end{enumerate}
\end{defi}

To show Lemma~\ref{lemm:interval_propositional_inter}, we use some pathological property of the SDE:
Under Assumption~\ref{cond:SDE_with_drift}, the SDE has the quite zig--zag sample path.
Then, once it reaches a region associated with an atomic proposition, it immediately hits the region with high frequency.

 \begin{lem}\label{lemm:local_extreme}
Let $X=\{X_t\}_{t\geq0}$ be the strong solution of the SDE \eqref{eq:SDE_with_drift} satisfying Assumption~\ref{cond:SDE_with_drift}. 
Then, the following statements hold (see~\ref{append:SDE} for the proof):
\begin{enumerate}[(i)]
\item Let $a\in\R$ and put 
    \begin{align}
        \mathcal{L}^a_\omega:=\{t\geq0; X_t(\omega)=a\},\hspace{10pt}\ \omega\in\Omega.
    \end{align}
    Then $\mathcal{L}^a_\omega$ is dense--in--itself almost surely.
    \item For almost every $\omega\in\Omega$, the set of points of local maximum and local minimum for the path $t\mapsto X(\omega)$ is dense in $[0,\infty)$, and all local maxima and local minima are strict.
\end{enumerate}
 \end{lem}

\begin{lem}\label{lemm:propositional_inter}
Let $X=\{X_t\}_{t\geq0}$ be the strong solution of the SDE~\eqref{eq:SDE_with_drift} satisfying Assumption~\ref{cond:SDE_with_drift}.
Let us define an atomic formula $a$ as $X(\omega),t\models a\Leftrightarrow X_t(\omega)\in \langle y,\infty)$, where $\langle y,\infty) $ is half-line with open or closed endpoint $y\in\R$.
Then $a$ satisfies \eqref{eq:phi_closure_of_int} and \eqref{eq:notphi_closure_of_int} almost surely.
Namely,
\begin{align*}
    \bracketo{a}&\subset\overline{\internal\bracketo{a}},\asure,\\
    \bracketo{\lnot a}&\subset\overline{\internal\bracketo{\lnot a}},\asure
\end{align*}
\end{lem}

\begin{proof}
    Put $\tilde{\Omega}$ be the set of $\omega\in\Omega$ with the following properties:
    \begin{enumerate}[(i)]
        \item The map $t\mapsto X_t(\omega)$ is continuous,
        \item $\mathcal{L}^y_\omega:=\{t\geq0 ; X_t(\omega)=y\}$ is dense--in--itself,
        \item  the set of local maximum and local minimum of $t\mapsto X_t(\omega)$  is dense in $[0,\infty)$, and 
        \item  all the local minima and the local maxima are strict
    \end{enumerate}
    Then there exists some $\hat{\Omega}\in\F$ such that $\hat{\Omega}\subset\tilde{\Omega}$ and $\Prob(\hat{\Omega})=1$ because of Definition~\ref{defi:SDE_with_drift}, Lemma~\ref{lemm:local_extreme}.
    Indeed, let $\Omega_1,\Omega_1,\Omega_3,\Omega_4\in\F$ be the sets of $\omega$ such that (i)--(iv) holds respectively and $\Prob(\Omega_1)=\Prob(\Omega_2)=\Prob(\Omega_3)=\Prob(\Omega_4)=1$, then $\Prob(\bigcap_{i=1}^4\Omega_i)=1$ follows from Remark~\ref{rema:almost_sure_concerves}.
    From now on, let us prove that the formula $a$ satisfies \eqref{eq:phi_closure_of_int} and \eqref{eq:notphi_closure_of_int} for all $\omega\in\hat{\Omega}$.
    
        Let $\langle y,\infty)$ be the left-closed interval $[y,\infty)$.
    The statement $t\in\llbracket \lnot a\rrbracket_\omega$ is equivalent to $X_t(\omega)<y$.
    Since $t\mapsto X_t(\omega)$ is continuous, the set
    \begin{align*}
        \bracketo{\lnot a}=\{t\geq0 ;X_t(\omega)<y\}
    \end{align*}
    is an open set, and therefore inclusion \eqref{eq:notphi_closure_of_int} holds clearly.
    On the other hand, due to the continuity of $t\mapsto X_t(\omega)$, if $X_t(\omega)>y$, then it implies that $t\in\text{ int }\llbracket a\rrbracket_\omega$.
    Hence, it remains to show that $X_t(\omega)=y$ implies $t\in\overline{\internal\bracketo{a}}$.
    Suppose $X_t(\omega)=y$ and $t\notin\overline{\text{int}{\llbracket a\rrbracket_\omega}}$. Then, there exists $\varepsilon>0$ such that $(t-\varepsilon,t+\varepsilon)\cap\text{int}\llbracket a\rrbracket_\omega=\emptyset$.
    Since $(\exists s\in(t-\varepsilon,t+\varepsilon))[X_s(\omega)>y]$ implies $(t-\varepsilon,t+\varepsilon)\cap\text{int}\llbracket a\rrbracket_\omega\neq\emptyset$, it follows that $(\forall s\in(t-\varepsilon,t+\varepsilon))[X_s(\omega)\leq y]$.
    By applying (iii), we can conclude that $t$ is a strict local maximum, i.e., $(\forall s\in(t-\varepsilon,t+\varepsilon)\setminus\{t\})[X_s(\omega)<y]$, and thus, $t$ is an isolated point of $\{t\geq0; X_t(\omega)=y\}$.
    However, this contradicts (ii).
    Therefore, we obtain the inclusion \eqref{eq:phi_closure_of_int}.
    
    On the other hand, consider $\langle y,\infty)$ as the left-open interval $(y,\infty)$.
    Now, $t\in\llbracket a\rrbracket_\omega$ is equivalent to $X_t(\omega)>y$.
    Since $t\mapsto X_t(\omega)$ is continuous, the set
    \begin{align*}
        \bracketo{a}=\{t\geq0 ; X_t(\omega)>y\}
    \end{align*}
    is an open set, and thus inclusion \eqref{eq:phi_closure_of_int} holds clearly.
    Moreover, due to the continuity of $t\mapsto X_t(\omega)$, if $X_t(\omega)<y$, then it implies $t\in\text{ int }\llbracket\lnot a\rrbracket_\omega$.
    Therefore, we need to show the inclusion \eqref{eq:notphi_closure_of_int} when $t\in\llbracket \lnot a \rrbracket_\omega$ and $X_t(\omega)=y$.
    Suppose $t\notin\overline{\text{int}{\llbracket\lnot a\rrbracket_\omega}}$, which implies the existence of $\varepsilon>0$ such that $(t-\varepsilon,t+\varepsilon)\cap\text{int}\llbracket \lnot a\rrbracket_\omega=\emptyset$.
    If $(\exists s\in(t-\varepsilon,t+\varepsilon))[X_s(\omega)<y]$, then it implies $(t-\varepsilon,t+\varepsilon)\cap\text{int}\llbracket \lnot a\rrbracket_\omega\neq\emptyset$. Consequently, it holds that $(\forall s\in(t-\varepsilon,t+\varepsilon))[X_s(\omega)\geq y]$.
    By applying~\ref{lemm:local_extreme}--(ii), we can deduce that $t$ is a strict local minimum, i.e., $(\forall s\in(t-\varepsilon,t+\varepsilon)\setminus\{t\})[X_s(\omega)>y]$, which means $t$ is an isolated point of $\llbracket \lnot a\rrbracket_\omega$.
    However, this contradicts (ii).
    Thus, we establish \eqref{eq:notphi_closure_of_int}.        
\end{proof}

\begin{proof}[Proof of Lemma~\ref{lemm:interval_propositional_inter}]
    Let us define atomic formulas $a,b$ as
    \begin{align}
        X(\omega),t\models a&\Leftrightarrow X_t(\omega)\in\langle y_1,\infty),\label{eq:interval_a}\\
        X(\omega),t\models b&\Leftrightarrow X_t(\omega)\in\langle y_2,\infty)\label{eq:interval_b},
    \end{align}
    where left endpoints $y_1,y_2$ can be open or closed and satisfy $y_1<y_2$.
    Then we can define $X(\omega),t\models p$ is equivalent to $X(\omega),t\models a\wedge \lnot b$ and hence it is enough to show that $a\wedge \lnot b$ satisfies \eqref{eq:phi_closure_of_int} and \eqref{eq:notphi_closure_of_int} almost surely.
    To see this, let $\tilde{\Omega}$ be the set of $\omega\in\Omega$ with the following properties
    \begin{enumerate}[(i)]
        \item[(i)] The map $t\mapsto X_t(\omega)$ is continuous,
        \item[(ii)] the formula $a$ satisfies \eqref{eq:phi_closure_of_int} and \eqref{eq:notphi_closure_of_int}, and
        \item[(iii)] the formula $b$ satisfies \eqref{eq:phi_closure_of_int} and \eqref{eq:notphi_closure_of_int}. 
    \end{enumerate}
    Then Definition~\ref{defi:SDE_with_drift} and Lemma~\ref{lemm:propositional_inter} imply that there exists some $\hat{\Omega}\in\F$ such that $\Prob(\hat{\Omega})=1$ and every $\omega\in\hat{\Omega}$ satisfies (i)--(iii). Lemma~\ref{lemm:propositional_inter}.
    Now let us show \eqref{eq:phi_closure_of_int} and \eqref{eq:notphi_closure_of_int} for every $\omega\in\hat{\Omega}$.
    \begin{description}
        \item[\eqref{eq:phi_closure_of_int}] Given that
        \begin{align*}
             \bracketo{a\wedge \lnot b}\subset(\internal\bracketo{a}\cap\internal\bracketo{\lnot b})\cup\partial\bracketo{a}\cup\partial\bracketo{\lnot b},
         \end{align*}
         and 
         \begin{align*}
            \text{int}\llbracket a\rrbracket_\omega\cap\internal\llbracket \lnot b\rrbracket_\omega=\text{int}\llbracket a\wedge \lnot b\rrbracket_\omega\subset\overline{\internal\bracketo{a\wedge\lnot b}},
        \end{align*}
         then it is enough to show
         \begin{align*}
             &\bracketo{a\wedge\lnot b}\cap\partial\llbracket a\rrbracket_\omega\subset \overline{\internal\bracketo{a\wedge\lnot b}}\\
             &\bracketo{a\wedge\lnot b}\cap\partial\bracketo{\lnot b}\subset\overline{\internal\bracketo{a\wedge\lnot b}}.
         \end{align*}
        
        Suppose that $t\in\llbracket a\wedge\lnot b\rrbracket_\omega\cap\partial{\bracketo{a}}$.
        If $O$ is a neighborhood of $t$,  $O\cap\internal\bracketo{a}\neq\emptyset$ because $O\cap\bracketo{a}\neq\emptyset$ and (ii) hold.
        Since the path $t\mapsto X(\omega)$ is continuous,  $t\in \partial\bracketo{a}$ implies $X_t(\omega)=y_1<y_2$ and hence $t\in\internal\bracketo{\lnot b}$.
        Let $\varepsilon>0$.
        Since $(t-\varepsilon,t+\varepsilon)\cap\internal\bracketo{\lnot b}$ is a neighborhood of $t$, 
     \begin{align*}
         (t-\varepsilon,t+\varepsilon)\cap\internal\bracketo{a}\cap\internal\bracketo{\lnot b}=(t-\varepsilon,t+\varepsilon)\cap\internal\bracketo{a\wedge\lnot b}\neq\emptyset,\asure,
     \end{align*}
     and hence $t\in\overline{\internal\bracketo{a\wedge\lnot b}}$.
        
        The same argument can be applied when $t\in\bracketo{a\wedge\lnot b}\cap\partial\bracketo{\lnot b}$.
        Thus, we have shown \eqref{eq:phi_closure_of_int}.
        \item[\eqref{eq:notphi_closure_of_int}] 
         Suppose that $t\in\bracketo{\lnot a\vee b}$ and let $O$ be a neighborhood of $t$.
        Since $t\in\bracketo{\lnot a}$ or $t\in\bracketo{b}$, (ii) and (iii) imply that $O\cap\internal\bracketo{\lnot a}\neq\emptyset$ or $O\cap \internal\bracketo{b}\neq\emptyset$ holds.
        Since $\internal\bracketo{\lnot a}\cup\internal\bracketo{b}\subset\internal(\bracketo{\lnot a}\cup\bracketo{b})$, it holds that $O\cap\internal\bracketo{\lnot a\vee b}=O\cap\internal(\bracketo{\lnot a}\cup\bracketo{b})\neq\emptyset$.
        Then $t\in\overline{\internal\bracketo{\lnot a\vee b}}$.
        \qedhere
    \end{description}        
\end{proof}

We can interpret the boundary $\partial\bracketo{\phi}$ of time set $\bracketo{\phi}$ as the time that the indicator function $\chi_\phi(\omega,t)$ in Definition~\ref{defi:indicator_func} changes its value.
The following lemma shows that the boundary $\partial \bracketo{\phi}$ of every MTL--formula $\phi$ has Lebesgue measure zero almost surely if $X_t$ has a density for all $t>0$ and $AP$ is distinct in the sense of the following definition.
\begin{defi}
A Borel set $B$ on $\R$ is \emph{distinct} if its boundary $\partial B$ has Lebesgue measure zero.
    An atomic formula $a\in AP$ is \emph{distinct} when the corresponding set $B_a$ is distinct.
    The set of $AP$ of atomic formulas is said to be \emph{distinct} when all $a\in AP$ are distinct.
\end{defi}
\begin{lem}\label{lemm:closedboundary_clt_interval}
	Consider the case of the distinct set $AP$ of atomic formulas.
	Let $(\Omega,\F,\Prob)$ be a complete probability space.
	Suppose that $X$ is an almost surely continuous stochastic process such that $X_t$ has a density for every $t\in(0,\infty)$.
	Then, for every MTL--formula $\phi$ there exists some measurable set $K\in\F\otimes\B([0,\infty))$ such that 
	\begin{enumerate}[(i)]
		\item[(i)] $\{t;(\omega,t)\in K\}$ is almost surely closed,
		\item[(ii)] $\{t;(\omega,t)\in K\}$ has Lebesgue measure zero almost surely, 
		\item[(iii)] $\Prob(\{\omega;(\omega,t)\in K\})=0$ for every $t\in(0,\infty)$, and
		\item[(iv)] $\partial\bracketo{\phi}\subset\{t;(\omega,t)\in K\}$.
	\end{enumerate}
\end{lem}
For this proposal, in the following lemma, we show that the boundary of the time set of the form $\Diamond_{\lrangle{S,T}}\phi$ is restricted to the shift of the boundary of the form $\phi$.
\begin{lem}\label{lemm:boundformula_clt}
    Let $\phi$ be an MTL--formula and $\langle S,T\rangle$ be a positive interval on $[0,\infty)$.
    Then it holds almost surely that 
    \begin{align}
	\partial \llbracket\Diamond_{\langle S,T\rangle}\phi\rrbracket_\omega\subset [(\partial\llbracket\phi\rrbracket_\omega\ominus S)\cup(\partial\llbracket\phi\rrbracket_\omega\ominus T)],
    \end{align}
    where $\partial\llbracket\phi\rrbracket_\omega\ominus S:=\{t-S;t\in \partial\llbracket\phi\rrbracket_\omega\}\cap[0,\infty)$ and $\partial\llbracket\phi\rrbracket_\omega\ominus T:=\{t-T;t\in \partial\llbracket\phi\rrbracket_\omega\}\cap[0,\infty)$.
\end{lem}
\begin{proof}
    Let $\langle S,T\rangle$ be the closed interval $[S,T]$.
    Suppose that $t\in\partial \llbracket\Diamond_{\langle S,T\rangle}\phi\rrbracket_\omega$.
    Then it is clear that $(t+S,t+T)\cap\llbracket\phi\rrbracket_\omega=\emptyset$.
    If not, there exists some neighborhood of $t$ whose every element $s$ satisfies $(s+S,s+T)\cap\llbracket\phi\rrbracket_\omega\neq\emptyset$ and hence $t\notin\partial \llbracket\Diamond_{\langle S,T\rangle}\phi\rrbracket_\omega$.
    Again from $t\in\partial \llbracket\Diamond_{\langle S,T\rangle}\phi\rrbracket_\omega$, one of the following two statement holds:
    \begin{enumerate}[(i)]
        \item[(i)] There exist some sequence $t_n,\ n=1,2,3,\cdots$ in $\llbracket\Diamond_{[S,T]}\phi\rrbracket_\omega$ such that $\sup_{n}t_n=t$.
        \item[(ii)] There exist some sequence $t_n,\ n=1,2,3,\cdots$ in $\llbracket\Diamond_{[S,T]}\phi\rrbracket_\omega$ such that $\inf_{n}t_n=t$.
    \end{enumerate}
    It is enough to show $S+t\in\partial\bracketo{\phi}$ or $T+t\in\partial\bracketo{\phi}$ for (i) and (ii).
    \begin{enumerate}[(i)]
        \item Since $(S+t,T+t)\subset\llbracket\lnot\phi\rrbracket_\omega$, $(S+t,S+t+\varepsilon)\cap\llbracket\lnot\phi\rrbracket_\omega\neq\emptyset$ for every positive $\varepsilon$.
        Together with $\llbracket \phi\rrbracket_\omega\cap [S+t_n,T+t_n]\neq\emptyset$, $(S+t,T+t)\subset\llbracket\lnot\phi\rrbracket_\omega$ also implies $[S+t_n,S+t]\cap\llbracket\phi\rrbracket_\omega\neq\emptyset$ for every $n\in\N$.
        Then $(S+t-\varepsilon,S+t]\cap\llbracket\phi\rrbracket_\omega\neq\emptyset$ for every positive $\varepsilon$.
        \item We can show $(T+t-\varepsilon,T+t)\cap\llbracket\lnot\phi\rrbracket_\omega\neq\emptyset$ and $[T+t,T+t+\varepsilon)\cap\llbracket\phi\rrbracket_\omega\neq\emptyset$ by showing $[T+t,T+t_n]\cap\llbracket\phi\rrbracket_\omega\neq\emptyset$.
        Indeed, $(S+t,T+t)\cap\llbracket\phi\rrbracket_\omega=\emptyset$ and $[S+t_n,T+t_n]\cap\llbracket\phi\rrbracket_\omega\neq\emptyset$ implies $[T+t,T+t_n]\cap\llbracket\phi\rrbracket_\omega\neq\emptyset$.
    \end{enumerate}
    Thus, we show the statement when $\langle S,T\rangle$ is closed.
    We can prove the case of $\langle a,b\rangle=(S,T),[S,T),[S,T)$ in exactly the same way.        
\end{proof}

\begin{proof}[Proof of Lemma~\ref{lemm:closedboundary_clt_interval}]
    Since the map $t\mapsto X_t(\omega)$ is almost surely continuous, there exists some $N\in\F$ such that
    $\Prob(N)=0$ and $t\mapsto X_t(\omega)$ is continuous whenever $\omega\notin N$.
    Suppose $a$ is an atomic formula and $t\in\partial\llbracket a\rrbracket_\omega$. For any positive $\varepsilon$, we can find $s$ and $s'$ in the interval $(t-\varepsilon,t+\varepsilon)$ such that $X_s(\omega)\in B_a$ and $X_{s'}(\omega)\notin B_a$. This is because $t$ is a boundary point of the satisfaction set $\llbracket a\rrbracket_\omega$.
    Since the mapping $t\mapsto X_t(\omega)$ is continuous for $\omega\notin N$, it follows that $X_t(\omega)$ lies on the boundary $\partial B_a$.
    In other words, $X_t(\omega)$ is located on the boundary of the set defined by the atomic formula $a$.
    Put $K:=\{(\omega,t); X_t(\omega) \in\partial B_a \}\cup N\times[0,\infty)$ and $K_\omega:=\{t ; (\omega,t)\in K\}$.
    Hence, we get $\partial \llbracket a\rrbracket_\omega\subset K_\omega$.
    Since $t\mapsto X_t(\omega)$ is continuous almost surely and $X_t$ has a density for every $t>0$, $K$ is measurable, $K_\omega$ is almost surely closed, 
    \begin{align*}
        \Prob(\{\omega;(\omega,t)\in K\})\leq\Prob(\omega ;X_t(\omega)\in\partial B_a)+\Prob(N)=0\hspace{10pt}\forall t\in (0,\infty).
    \end{align*}
    Then it holds that
    \begin{align}
        \int_{[0,\infty)}\left\{\int_\Omega \1_K(\omega,t) \Prob(d\omega)\right\} dt=0.\label{eq:closezero}
    \end{align}
    By using Fubini's Theorem (see Theorem 8.8 in \cite{ouc.200337561119660101}), we have
    \begin{align*}
        \int_\Omega\left\{\int_{[0,\infty)} \1_K(\omega,t) dt\right\} \Prob(d\omega)=0,
    \end{align*}
    which implies that $K_\omega$ has Lebesgue measure zero almost surely (see (b) of Theorem 1.39 in \cite{ouc.200337561119660101}).
    When $K$ corresponds to a formula $\phi$ with (i)-(iv), then $K$ also satisfies (i)-(iv) for$\lnot\phi$, since $\partial\llbracket \lnot \phi\rrbracket_\omega =\partial\llbracket \phi\rrbracket_\omega$.
    When $K_1$ and $K_2$ satisfy (i)-(iv) for $\phi_1$ and $\phi_2$ respectively, $K_1\cup K_2$ satisfies (i)-(iv) for $\phi_1\wedge\phi_2$, since $\{t;(\omega,t)\in K_1\cup K_2\}$ is closed, $\Prob(\omega; (\omega,t)\in K_1\cup K_2)=0$ for $t\in(0,\infty)$, and $\partial \llbracket \phi_1\wedge \phi_2\rrbracket_\omega\subset \{t;(\omega,t)\in K_1\cup K_2\}$.
    Suppose that $K$ satisfies (i)-(iv) for $\phi$.
    We show that $\{(\omega,t); t\in[(K_\omega\ominus S)\cup(K_\omega\ominus T)]\}$ satisfies (i)--(iv) for $\Diamond_{\langle S,T\rangle}\phi$.
    \begin{enumerate}[(i)]
        \item Since $K_\omega$ is closed almost surely, $(K_\omega\ominus S)$ and $(K_\omega\ominus T)$ are almost surely closed and then $(K_\omega\ominus S)\cup(K_\omega\ominus T)$ is closed almost surely.
        \item From \eqref{eq:closezero}, it holds that 
    \begin{align}
        \int_{[0,\infty)}\left\{\int_\Omega \1_{\{t\in K_\omega\ominus S\}}(\omega,t) \Prob(d\omega)\right\} dt=\int_{[S,\infty)}\left\{\int_\Omega \1_{\{t\in K_\omega\}}(\omega,t) \Prob(d\omega) \right\} dt=0,\nonumber\\
        \int_{[0,\infty)}\left\{\int_\Omega \1_{\{t\in K_\omega\ominus T\}}(\omega,t)\ \Prob(d\omega)\right\} dt=\int_{[T,\infty)}\left\{\int_\Omega \1_{\{t\in K_\omega\}}(\omega,t) \Prob(d\omega) \right\} dt=0.\nonumber
    \end{align}
    By Fubini's theorem, we have
    \begin{align*}
        \int_\Omega\left\{ \int_{[0,\infty)}\1_{\{t\in K_\omega\ominus S\}}(\omega,t) dt\right\} \Prob(d\omega)&=0,\\
        \int_\Omega\left\{ \int_{[0,\infty)}\1_{\{t\in K_\omega\ominus T\}}(\omega,t) dt\right\} \Prob(d\omega)&=0,
    \end{align*}
    which implies that $\{(\omega,t); t\in(K_\omega\ominus S)\cup(K_\omega\ominus T) \}$ has Lebesgue measure zero almost surely.
    \item When $t>0$, we have
    \begin{align*}
        &\Prob(\omega ; t\in(K_\omega\ominus S)\cup(K_\omega\ominus T))\\
        \leq&\Prob(\omega ; t\in(K_\omega\ominus S))+\Prob(\omega ;t\in (K_\omega\ominus T))\\
        \leq&\Prob(\omega ; t+S\in K_\omega)+\Prob(\omega ;t+T\in K_\omega)=0.
    \end{align*}
    \item From Lemma~\ref{lemm:boundformula_clt}, $\partial \llbracket\Diamond_{\langle S,T\rangle}\phi\rrbracket_\omega\subset [(\partial\llbracket\phi\rrbracket_\omega\ominus S)\cup(\partial\llbracket\phi\rrbracket_\omega\ominus T)]\subset [(K_\omega\ominus S)\cup(K_\omega\ominus T)]$ almost surely. \qedhere
    \end{enumerate}        
\end{proof}

In the following lemma, we give a sufficient condition for convergence of the indicator function of the formula with diamond or box operator.
\begin{lem}\label{lemm:indicator_convergence}
Let $X$ be the solution of SDE \eqref{eq:SDE_with_drift} satisfying Assumption~\ref{cond:SDE_with_drift}.
Define an MTL--formula $p$ as 
\begin{align*}
    X(\omega),t\models p\Leftrightarrow X_t(\omega)\in B_p
\end{align*}
 for some distinct set $B_p$ on $\R$.
 Let $\langle S,T\rangle$ be a positive interval on $[0,\infty)$.
 If $p$ satisfies \eqref{eq:phi_closure_of_int} and \eqref{eq:notphi_closure_of_int}, the following statements hold:
\begin{enumerate}[(i)]
    \item Define $\phi:=\Diamond_{\langle S,T\rangle}p$.
    Then $\chi^{(n)}_{\phi}(\omega,t)\rightarrow\chi_{\phi}(\omega,t)$ for every $t\in[0,\infty)$.
    \item Define $\psi:=\Box_{\lrangle{S,T}}p$.
    Then $\chi^{(n)}_{\psi}(\omega,t)\rightarrow\chi_{\psi}(\omega,t)$ for every $t\in[0,\infty)$.
\end{enumerate}
Here, $\chi_{\phi}^{(n)}(\omega,t)$, $\chi_{\psi}^{(n)}(\omega,t)$, $\chi_{\phi}(\omega,t)$, and $\chi_{\psi}(\omega,t)$ are the indicator functions defined in Definition~\ref{defi:indicator_func}.
\end{lem}
\begin{proof}
    First, let us show that $\lrangle{t+S,t+T}\cap\bracketo{p}\neq\emptyset$ implies $(t+S,t+T)\cap\internal\bracketo{p}\neq\emptyset$ almost surely.
If $\partial\bracketo{p} \cap \langle t+S,t+T\rangle = \emptyset$ and $\lrangle{t+S,t+T}\cap\bracketo{p}\neq\emptyset$, then it follows that $X(\omega),s\models p$ for all $s\in\lrangle{t+S,t+T}$, and therefore, $\lrangle{t+S,t+T}\subset\bracketo{p}$. Consequently, $(t+S,t+T)\cap\internal\bracketo{p}\neq\emptyset$.
Next, suppose $\partial\bracketo{p} \cap \lrangle{t+S,t+T} \neq \emptyset$ and $\lrangle{t+S,t+T}\cap\bracketo{p}\neq\emptyset$.
Since $\{X_t\}_{t\geq0}$ satisfies Assumption~\ref{cond:SDE_with_drift}, $X_t$ has a density for $t>0$ by~\ref{prop:density_SDE}.
When $t+S>0$, we have $t+T>t+S>0$, and Lemma~\ref{lemm:closedboundary_clt_interval} implies that $t+S$ and $t+T$ do not belong to $\partial\bracketo{p}$ almost surely.
Thus, we have $\partial\bracketo{p}\cap(t+S,t+T)\neq\emptyset$, which implies $\bracketo{p}\cap(t+S,t+T)\neq\emptyset$.
Therefore, we conclude from \eqref{eq:phi_closure_of_int} that $(t+S,t+T)\cap\internal\bracketo{p}\neq\emptyset$.
If $t+S=0$, $\partial \bracketo{p}$ intersects the open set of the form $[0,t+T)$ or $(0,t+T)$ on $[0,\infty)$.
Then, from \eqref{eq:phi_closure_of_int}, we can conclude
$\internal\bracketo{p}\cap(t+S,t+T)\neq\emptyset$.

We can show in similar way that $\lrangle{t+S,t+T}\cap\bracketo{\lnot p}\neq\emptyset$ implies $(t+S,t+T)\cap\internal\bracketo{\lnot p}\neq\emptyset$ almost surely.

Now let us prove (i) and (ii).
Suppose $X(\omega),t\models\phi$.
Since $(t+S,t+T)\cap\internal\bracketo{p}$ is a nonempty open set, there exists $s\in(\Lambda_n(t)+S,\Lambda_n(t)+T)\cap\N/n$ such that $X(\omega),s\models_np$ for sufficiently large $n$.
Hence, $X,\Lambda_n(t)\models_n\phi$.
By applying the same argument, we can show from \eqref{eq:notphi_closure_of_int} that if $X(\omega),t\not\models\psi$, then $X(\omega),\Lambda_n(t)\not\models_n\psi$ for sufficiently large $n$.

On the other hand, suppose $X(\omega),t\not\models\phi$.
Then $\bracketo{p}\cap(S,T)=\emptyset$ and $\partial\bracketo{p}\cap(S,T)=\emptyset$.
If $t+S>0$, according to~\ref{prop:density_SDE} and Lemma~\ref{lemm:closedboundary_clt_interval}, $t+S$ and $t+T$ do not belong to $\partial\bracketo{p}$ almost surely.
Thus, there exists $\varepsilon>0$ such that $(t+S-\varepsilon,t+T-\varepsilon)\subset\bracketo{\lnot p}$, and hence $(\Lambda_n(t)+S,\Lambda_n(t)+T)\cap\bracketo{p}=\emptyset$ for sufficiently large $n$.
If $t+S=0$, since $\Lambda_n(t)=t=0$, it holds that
\begin{align*}
&X(\omega),\Lambda_n(t)\not\models_n \phi\Leftrightarrow X(\omega),0\not\models_n \Diamond_{\lrangle{0,T}}p,\\
&X(\omega),t\not\models \phi\Leftrightarrow X(\omega),0\not\models \Diamond_{\lrangle{0,T}}p.
\end{align*}
Then it is clear that $X(\omega),t\not\models \phi$ implies $X(\omega),\Lambda_n(t)\not\models_n \phi$.
Now we have shown that $X(\omega),t\not\models_n\phi$ for sufficiently large $n$.
The same argument can be applied to prove that if $X(\omega),t\models\psi$, then $X(\omega),\Lambda_n(t)\models_n\psi$ for sufficiently large $n$.
\end{proof}

\begin{lem}\label{lemm:propositional_denseness}
    Suppose that a propositional formula $p$ satisfies the conditions introduced in the statement of Theorem~\ref{theo:pairwise_disjoint_diamond}.
    Specifically, let $\lrangle{x_i,y_i},\ i=1,\cdots,k$ be pairwise disjoint positive intervals, and define $B_p:=\bigcup_{i=1}^k\lrangle{x_i,y_i}$.
    Define a propositional formula $p,p_1,\cdots,p_k$ by 
    \begin{align*}
        &X(\omega),t\models p\Leftrightarrow X_t(\omega)\in B_p,\\
        &X(\omega),t\models_n p\Leftrightarrow X_t(\omega)\in B_p,\\
        &X(\omega),t\models p_i\Leftrightarrow X_t(\omega)\in \lrangle{x_i,y_i} \quad \text{for }i=1\cdots,k,\\
        &X(\omega),t\models_n p_i\Leftrightarrow X_t(\omega)\in \lrangle{x_i,y_i} \quad \text{for }i=1\cdots,k.
    \end{align*}
    
    Then $p$ satisfies \eqref{eq:phi_closure_of_int} and \eqref{eq:notphi_closure_of_int}.
    Namely,
    \begin{align*}
        \bracketo{p}&\subset\overline{\internal\bracketo{p}},\asure,\\
        \bracketo{\lnot p}&\subset\overline{\internal\bracketo{\lnot p}},\asure
    \end{align*}
\end{lem}

\begin{proof}
    First note that 
    \begin{align*}
        &X(\omega),t\models p\Leftrightarrow X(\omega),t\models \bigvee_{i=1}^k p_i,\\
        &X(\omega),t\models_n p\Leftrightarrow X(\omega),t\models_n\bigvee_{i=1}^k p_i,
    \end{align*}
    where $\bigvee_{i=1}^kp_i=p_1\vee p_2\vee\cdots\vee p_k$.
    If $B_p=\emptyset$ or $B_p=\R$, clearly $\bracketo{p}=\emptyset$ or $\bracketo{p}=[0,\infty)$, respectively.
    Hence \eqref{eq:phi_closure_of_int} and \eqref{eq:notphi_closure_of_int} holds.
    Otherwise, From Lemma~\ref{lemm:interval_propositional_inter}, every $p_i$ satisfies \eqref{eq:phi_closure_of_int} and \eqref{eq:notphi_closure_of_int} almost surely.
     Now we show that $\bracketo{p}(=\bracketo{\bigvee_{i=1}^kp_i})$ satisfies \eqref{eq:phi_closure_of_int} and \eqref{eq:notphi_closure_of_int}.
    \begin{description}
        \item[\eqref{eq:phi_closure_of_int}] Let $t\in\bracketo{\bigvee_{i=1}^kp_i}$ and $O$ be a neighborhood of $t$.
        Since $\bracketo{\bigvee_{i=1}^kp_i}=\bigcup_{i=1}^k\bracketo{p_i}$, there exists some $i\in\{1,\cdots,k\}$ such that $t\in\bracketo{p_i}$.
         Since $\bracketo{p_i}$ satisfies \eqref{eq:phi_closure_of_int} almost surely and $\internal\bracketo{p_i}\subset\internal\bracketo{\bigvee_{i=1}^kp_i}$, $O\cap\internal\bracketo{\bigvee_{i=1}^kp_i}\neq\emptyset$ almost surely.
        Then $\bigvee_{i=1}^kp_i$ satisfies \eqref{eq:phi_closure_of_int} almost surely.
        \item[\eqref{eq:notphi_closure_of_int}] Let $t\in\bracketo{\lnot \bigvee_{i=1}^kp_i}$.
    If $t\in\internal\bracketo{\lnot \bigvee_{i=1}^kp_i}$, then for any neighborhood $O$ of $t$, we have $O\cap\internal\bracketo{\lnot \bigvee_{i=1}^kp_i}\neq\emptyset$.
    Thus, it suffices to show that $O\cap\internal\bracketo{\lnot\bigvee_{i=1}^kp_i}\neq\emptyset$ for any neighborhood $O$ of $t$ whenever $t\in\partial\bracketo{\lnot \bigvee_{i=1}^kp_i}$.
    Since $\internal\bracketo{\lnot \bigvee_{i=1}^kp_i}=\internal(\bigcap_{i=1}^k\bracketo{\lnot p_i})=\bigcap_{i=1}^k\internal\bracketo{\lnot p_i}$, there must exist some $i\in{1,\cdots,k}$ such that $t\in\partial\bracketo{\lnot p_i}$.
    Indeed, if $t\in\internal\bracketo{\lnot p_i}$ for every $i$, then $t\in\internal\bracketo{\lnot\bigvee_{i=1}^k p_i}$.
    Since $t\mapsto X_t(\omega)$ is continuous almost surely and $\lrangle{x_1,y_1},\cdots\lrangle{x_k,y_k}$ are pairwise separated, we have $X_t(\omega)\in [x_j,y_j]^C$ when $j\neq i$ almost surely, and hence $t\in\internal\bracketo{\lnot p_j}$ for $i\neq j$ almost surely.
    Therefore, $t\in \bigcap_{j\neq i}\internal\bracketo{\lnot p_j}$.
    Then, $(t-\delta,t+\delta)\cap[0,\infty)\subset\bigcap_{j\neq i}\internal\bracketo{\lnot p_j}$ for sufficiently small $\delta>0$.
    Now, since $p_i$ satisfies \eqref{eq:phi_closure_of_int}, we have $(t-\delta,t+\delta)\cap\internal\bracketo{\lnot p_i}\neq\emptyset$.
    Hence, $(t-\delta,t+\delta)\cap\internal\bracketo{\bigvee_{j=1}^k\lnot p_j}=(t-\delta,t+\delta)\cap{\bigcap_{j=1}^k\internal\bracketo{\lnot p_j}}=(t-\delta,t+\delta)\cap\internal\bracketo{\lnot p_i}\neq\emptyset$. \qedhere
    \end{description}        
\end{proof}

\begin{proof}[Proof of Theorem~\ref{theo:pairwise_disjoint_diamond}]
    From the condition on $B_p$, it is a distinct set.
    Then Lemma~\ref{lemm:indicator_convergence} and Lemma~\ref{lemm:propositional_denseness} implies the almost sure convergence of $\chi^{(n)}{\phi}(\omega,t)$ and $\chi^{(n)}{\psi}(\omega,t)$ for every $t\in[0,\infty)$.
    Finally, Lemma~\ref{lemm:prob_convergence_from_indicator} can be employed to show the convergence of probability.        
\end{proof}

\subsection{\texorpdfstring{The case of $\flat$MTL--formulas}{The case of flat MTL-formulas}}\label{sec:general_flat_convergence}
Now, we prove the convergence result for general $\flat$MTL--formulas.
Let $X$ be the solution of SDE \eqref{eq:SDE_with_drift} with Assumption~\ref{cond:SDE_with_drift}.
Henceforth, we discuss under the following assumption:
\begin{asm}\label{cond:pairwise_disjoint}
    For every propositional formula $p$,
    \begin{align}
        X(\omega),t\models p\Leftrightarrow X_t(\omega)\in B_p,
    \end{align}
    for some $B_p$, which is a union of pairwise separated positive intervals on $\R$.
    Here $B_p$ may possibly be $\emptyset$ or $\R$.
\end{asm}
\begin{rem}
We give some examples of setting so that every propositional formula satisfies Assumption~\ref{cond:pairwise_disjoint}.
Let $B_1,\cdots,B_k$ be positive intervals on $\R$ such that
\begin{enumerate}[(i)]
    \item[(i)] $\bigcup_{i=1}^kB_i=\R$.
    \item[(ii)] $B_i\cap B_j=\emptyset$ if $i\neq j$.
\end{enumerate}
We define the semantics of atomic formulas $AP:=\{a_1,\cdots,a_k\}$ as $X(\omega),t\models a_i\Leftrightarrow X_t(\omega)\in B_i$ for $i=1,\cdots,k$.
Then, every propositional formula clearly satisfies Assumption~\ref{cond:pairwise_disjoint}.
Given this fact, if a propositional formula $p$ satisfies Assumption~\ref{cond:pairwise_disjoint}, $\lnot p$ also satisfies Assumption~\ref{cond:pairwise_disjoint}.
\end{rem}

Under these settings, we show the following statement.

\begin{thm}\label{theo:flat_MTL_convergence}
Suppose that $\{X_t\}_{t\geq0}$ is the solution of SDE \eqref{eq:SDE_with_drift} with Assumption~\ref{cond:SDE_with_drift}.
    Let $AP$ be the set of atomic formulas such that every propositional formula satisfies Assumption~\ref{cond:pairwise_disjoint}.
	Let $\phi$ be a $\flat$MTL--formula.
	Then $\chi^{(n)}_\phi(\omega,t)\rightarrow\chi_\phi(\omega,t)$ almost surely for every $t\in[0,\infty)$.
        In particular, $\Prob(\omega; X(\omega),\Lambda_n(t)\models_n\phi)\rightarrow\Prob(\omega; X(\omega),t\models\phi)$ for all $t\in[0,\infty)$.
\end{thm}

\begin{rem}
    Under Assumption~\ref{cond:pairwise_disjoint}, every $\flat$MTL--formula $\phi$ is the following form:
    \begin{align*}
        \phi = \bigodot_{i=1}^k \phi_i
    \end{align*}
    where $\bigodot_{i=1}^k:\{0,1\}^k\rightarrow \{0,1\}$ is a Boolean combination and every $\phi_i$ satisfies one of the following:
    \begin{itemize}
        \item $\phi_i$ is a propositional formula with Assumption~\ref{cond:pairwise_disjoint}.
        \item $\phi_i = \Diamond_I p$ with a positive interval $I$ and a propositional formula $p$ satisfying Assumption~\ref{cond:pairwise_disjoint}.
        \item $\phi_i = \Box_I p$ with a positive interval $I$ and a propositional formula $p$ satisfying Assumption~\ref{cond:pairwise_disjoint}.
    \end{itemize}
    In other words, $\flat$MTL is an MTL in which every temporal operator is constructed using a diamond operator, and no temporal operator is nested.
    On the other hand, the counterexample presented in Section~\ref{sec:discretization} is an MTL--formula with the quadruple nested diamond operator.
    Then, the convergence result for MTL--formulas with diamond operators depends on the nesting level.
\end{rem}

\begin{lem}\label{lemm:propositional_indicator_convergence}
    Put Assumption~\ref{cond:SDE_with_drift} and Assumption~\ref{cond:pairwise_disjoint}.
     Let $p$ be a propositional formula.
     Then $\chi^{(n)}_p(\omega,t)\rightarrow\chi_p(\omega,t)$ almost surely, for every $t\in(0,\infty)$.
     In particular, $\Prob(\omega; X(\omega),\Lambda_n(t)\models_np)\rightarrow\Prob(\omega; X(\omega),t\models p)$.
 \end{lem}
 
\begin{proof}
    First note that $X(\omega),t\models p$ is equivalent to $X_t(\omega)\in B_p$ for some $B_p\subset\R$.
    Let $t=0$. 
    Then $\Lambda_n(0)=0$ and hence $X(\omega),\Lambda_n(0)\models_n p$ is equivalent to $X(\omega),0\models B_p$.
    Next, let $t>0$.
   By the definition of indicator functions , $\chi_p(\omega,t)=1$ is equivalent to $X_t(\omega)\in B_p$ and $\chi^{(n)}_p(\omega,t)=1$ is equivalent to $X_{\Lambda_n(t)}(\omega)\in B_p$.
   From Assumption~\ref{cond:pairwise_disjoint}, $B_p$ is distinct, then Lemma~\ref{lemm:closedboundary_clt_interval} implies that $t\notin \partial\bracketo{p}$ almost surely.
   Then almost surely there exists some $\varepsilon>0$ such that $\chi_p(\omega,s)=\chi_p(\omega,t)$ for every $s\in (t-\varepsilon,t+\varepsilon)\cap[0,\infty)$.
   Then it holds almost surely that $\chi^{(n)}_p(\omega,t)=\chi_p(\omega,\Lambda_n(t))=\chi_p(\omega,t)$ for sufficiently large $n$.
\end{proof}

\begin{proof}[Proof of Theorem~\ref{theo:flat_MTL_convergence}]
    Fix $t\in[0,\infty)$.
    It is clear that $\phi$ is Boolean combination of $\{\phi_i,i=1,\cdots,k\}$, where $\phi_i$ is a propositional formula or formula of the form $\Diamond_{\lrangle{S,T}}p$ where $\lrangle{S,T}$ is positive interval and $p$ is propositional formula.
    Then there exists some function $\bigodot_{i=1}^k:\{0,1\}^k\longrightarrow\{0,1\}$ such that
    \begin{align}
        \chi_\phi(\omega,t)&=\bigodot_{i=1}^k\chi_{\phi_i}(\omega,t),\label{eq:general_flat_chi}\\
        \chi^{(n)}_\phi(\omega,t)&=\bigodot_{i=1}^k\chi^{(n)}_{\phi_i}(\omega,t).\label{eq:general_flat_chi_dis}
    \end{align}
    From Assumption~\ref{cond:pairwise_disjoint} and Lemma~\ref{lemm:propositional_denseness}, every propositional formula satisfies \eqref{eq:phi_closure_of_int} and \eqref{eq:notphi_closure_of_int}.
    Then we can apply Lemma~\ref{lemm:indicator_convergence} and Lemma~\ref{lemm:propositional_indicator_convergence} to show that $\chi^{(n)}_{\phi_i}(\omega,t)$ converges almost surely to $\chi_{\phi_i}(\omega,t)$ for every $i=1,\cdots,k$.
    Then, almost surely, there exists some large $N\in\N$ such that $\chi^{(n)}_{\phi_i}(\omega,t)=\chi_{\phi_i}(\omega,t)$ for $n\geq N$ and $i=1,\cdots,k$.
    Therefore the left--hand side of \eqref{eq:general_flat_chi_dis} converges to the left side of \eqref{eq:general_flat_chi} almost surely.
    Once we have shown the almost sure convergence of \eqref{eq:general_flat_chi_dis} to \eqref{eq:general_flat_chi}, one can apply Lemma~\ref{lemm:prob_convergence_from_indicator} to see the convergence of the probability.    
\end{proof}

\begin{rem}
    In Assumption~\ref{cond:pairwise_disjoint}, we represent every propositional formula as a union of pairwise separated sets.
    It is one of the vital assumptions in our result.
    Let $X:=\{X_t\}_{t\geq0}$ be one-dimensional Brownian motion starting at zero.
    Consider the case that $a,b$ are atomic formulas such that $X(\omega),\models a\Leftrightarrow X_t\in B_a$ and $X(\omega),t\models B_b$, where $B_a:=(-\infty,1)$ and $B_b:=(1,\infty)$.
    Let $p:=\lnot(a\wedge b)$ and $\phi:=\Diamond_{\lrangle{S,T}}p$.
    Since $\overline{B_a}\cap \overline{B_b}\neq\emptyset$, $B_a$ and $B_b$ is not separated.
    Since $X_0\neq1$ and $X_t$ has density for all $t>0$, $\Prob(\omega;X_t(\omega)=1)=0$ for all $t\in[0,\infty)\cap\Q$.
    Hence the sigma-additivity of probability measure implies $\Prob(\exists t\in[0,\infty)\cap\Q,\ X_t=1)=0$.
    Then $\Prob(\omega; X(\omega),\Lambda_n(0)\models_n \Diamond_{\lrangle{S,T}}p)=0$ for every $n,S,T$ and then $\Prob(\omega; X(\omega),\Lambda_n(0)\models_n \phi)=0$.
    On the other hand, if $\tau_1(\omega):=\inf\{t\geq0; X_t(\omega)=1\}\in(S,T)$ then $X(\omega),0\models\phi$.
    From Theorem~\ref{fact:hitting_time_density} and the monotonicity of probability measures, then $\Prob(\omega; X(\omega),0\models\phi)\geq\Prob(\tau_1\in(S,T))>0$ and then $\Prob(\omega; X(\omega),0\models_n\phi)$ does not converge to $\Prob(\omega; X(\omega),0\models\phi)$.
\end{rem}

\subsection{Example of convergence}
Let us show an example of an MTL--formula for standard Brownian motion to illustrate the convergence result Theorem~\ref{theo:flat_MTL_convergence}.
Consider the following settings:
\begin{itemize}
    \item $X := \{X_t\}_{t \geq 0}$ is a standard Brownian motion starting at $x \in \R$.
    \item $p$ is an atomic formula defined as
    \begin{align}
        &X(\omega), t \models p \iff X_t(\omega) \in (x, \infty),& t \in [0, \infty),\\
        &X(\omega), t \models_n p \iff X_t(\omega) \in (x, \infty),& t \in \N/n.
    \end{align}
    \item $\phi = \lnot (\Diamond_{(0, 1)} p)$.
\end{itemize}
It is well--known that the path of standard Brownian motion has a quite zig--zag sample path:
\begin{fact}[see 2.7.18 in \cite{MR1121940}]\label{fact:signBM}
	With probability one, the path $t \mapsto \{X_t(\omega) - x\}$ changes its sign infinitely many times in any time interval $[0,\epsilon],\ \epsilon>0$.
\end{fact}
Then there exists $t \in I$ such that $X_t(\omega) \in (x, \infty)$ almost surely, which implies 
\begin{align*}
    (\exists t \in (0,1)) [X(\omega), t \models p]
\end{align*}
almost surely.
Therefore we have $\Prob( \omega \in \Omega ; X(\omega), 0 \models \phi) = 0$.
On the other hand, the joint density function of $X$ has the following form (see Chapter~2 in \cite{MR1121940}):
\begin{align}
    &\Prob(X_{t_1} \in (-\infty, x], X_{t_2} \in (-\infty, x], \cdots, X_{t_n} \in (-\infty, x])\nonumber\\
    = &\int_{-\infty}^x \int_{-\infty}^x \cdots \int_{-\infty}^x p(t_1; x, y_1) p(t_2 - t_1; y_1, y_2) \cdots p(t_n - t_{n-1}; y_{n-1}, y_n)dy_1 dy_2\cdots dy_n,\label{eq:joint_function_BM}
\end{align}
where $t_1 < t_2 < \cdots < t_n$ and
\begin{align*}
    p(t, x, y) = \frac{1}{\sqrt{2 \pi t}}\exp\left\{-\frac{(x-y)^2}{2t} \right\}, \hspace{10pt} t > 0,\ x,y\in\R. 
\end{align*}
Since $p(t, x, y)$ is positive for all $t$, $x$, $y$, the probability \eqref{eq:joint_function_BM} is positive, namely,
    $\Prob(\omega \in \Omega ; X(\omega), 0 \models_n \phi) > 0$
for every $n \in \N$.
Despite its positivity, Theorem~\ref{theo:flat_MTL_convergence} assures that $\Prob(\omega \in \Omega ; X(\omega), 0 \models_n \phi)\rightarrow \Prob(\omega \in \Omega ; X(\omega), 0 \models \phi) = 0$ as $n \rightarrow \infty$.

\section{Conclusion}

In conclusion, this study has examined the measurability of events defined by continuous MTL--formulas, assuming that the underlying stochastic process is measurable as a mapping from sample to time.

Moreover, we demonstrated a counterexample that highlights the lack of convergence of the probability derived from discrete semantics to that derived from continuous semantics, specifically when the intervals within diamond operators are allowed to be bounded open.

Furthermore, the study explored the case of $\flat$MTL--formulas, which only have $\Box$ or $\Diamond$ without nest as modalities, and demonstrated that the probability obtained from discrete semantics converges to the probability obtained from continuous semantics for every formula within this framework.
This finding suggests that $\flat$MTL--formulas exhibit a desirable convergence property, highlighting their applicability and reliability in capturing system behaviors.

In light of these results, our future work should focus on understanding the underlying factors and mechanisms that contribute to the convergence or divergence of probability between discrete and continuous semantics in various formula contexts. By gaining deeper insights into these dynamics, we may enhance the effectiveness and accuracy of probability simulations and predictions within formal verification and system analysis.

\bibliographystyle{alphaurl}
\bibliography{paper}

\appendix

\section{Some appendices for stochastic calculus}\label{append:SDE}

In this section, we prove Lemma~\ref{lemm:local_extreme}.
Let us recall the claim:

\begin{lem}
Let $X$ be a strong solution of SDE \eqref{eq:SDE_with_drift} on $(\Omega,\F,\Prob)$.
Put Assumption~\ref{cond:SDE_with_drift}. 
Then, the following statements hold:
\begin{enumerate}[(i)]
\item Put 
    \begin{align}
        \mathcal{L}^a_\omega:=\{t\geq0; X_t(\omega)=a\},\hspace{10pt}a\in\R,\ \omega\in\Omega.
    \end{align}
    Then $\mathcal{L}^a_\omega$ is dense in itself, almost surely, for all $a\in\R$.
    \item Almost surely, the set of points of local maximum and local minimum for the path $t\mapsto X(\omega)$ is dense in $[0,\infty)$, and all local maxima and local minima are strict.
\end{enumerate}
 \end{lem}
 
Toward this goal, we cite a similar fact about Brownian motion:
\begin{thmC}[2.9.7 and 2.9.12 in \cite{MR1121940}]\label{fact:strictlocalmaxima}
    Let $X$ be a Brownian motion on $(\Omega,\F,\Prob)$.
    Then the following statements hold:
    \begin{enumerate}[(i)]
        \item Put 
    \begin{align}
        \mathcal{L}^a_\omega:=\{t\geq0; X_t(\omega)=a\},\hspace{10pt}a\in\R,\ \omega\in\Omega.
    \end{align}
    Then $\mathcal{L}^a_\omega$ is dense in itself, almost surely, for all $a\in\R$.
    \item Almost surely, the set of points of local maximum for the Brownian path is dense in $[0,\infty)$, and all local maxima are strict.
    \end{enumerate}
\end{thmC}

 \begin{lem}\label{fact:strictlocalminima}
Almost surely, the set of points of local minimum for the Brownian path $t\mapsto W_t(\omega)$ is dense in $[0,\infty)$, and all local minima are strict.
\end{lem}

\begin{proof}
    Let $X$ be Brownian motion and define $\tilde{X}:=\{\tilde{X}_t\}_{t\geq0}$ by $\tilde{X}_t(\omega):=-X_t(\omega)$.
	 By rotational invariance (see 3.3.18 in \cite{MR1121940}), $\tilde{X}$ is also Brownian motion. We can apply Theorem~\ref{fact:strictlocalmaxima} so that the set of local maximum for $\tilde{X}(\omega)$ is dense, and all the local maxima are strict, almost surely.
	 Now, since all the local minima of $X(\omega)$ are local maxima of $\tilde{X}(\omega)$, then our statement holds.
\end{proof}

Now, we have shown Lemma~\ref{lemm:local_extreme} for the case of Brownian motion.
It remains to extend this statement to the case of SDE \eqref{eq:SDE_with_drift} under Assumption~\ref{cond:SDE_with_drift}.
Before showing the solution of such an SDE, we have to guarantee the existence and uniqueness of SDE:

\begin{thmC}[5.5.17 in \cite{MR1121940}]\label{prop:strong_sol_of_SDE}
	Assume that $ b:\R\rightarrow\R$ is bounded and $\sigma:\R\rightarrow\R$ is Lipschitz continuous with $\sigma^2$ bounded away from zero on every compact subset of $\R$.
	Then, for every initial $\xi\in\R$, equation \eqref{eq:SDE_with_drift} has unique strong solution.
\end{thmC}

Furthermore, we show the convergence of the probability of MTL--formula for SDEs with a density function.
The following proposition assures the existence of density for SDE \eqref{eq:SDE_with_drift}:

\begin{thmC}[Theorem 2.1 in \cite{10.3150/09-BEJ215}]\label{prop:density_SDE}
	Let $\xi$ be the constant value in $\R$.
	Assume that $\sigma$ is H\"older continuous with exponent $\theta\in[1/2,1]$ and that $b$ is measurable and at most linear growth.
	Consider a continuous solution $\{X_t\}_{t\geq0}$ to \eqref{eq:SDE_with_drift}.
	Then, for all $t>0$, the law of $X_t$ has a density on the set $\{x\in\R; \sigma(x)\neq0\}$.
\end{thmC}
It is easy to make sure that the above two propositions can be applied to the unique existence and absolute continuity of $X$ under Assumption~\ref{cond:SDE_with_drift}.

To extend Theorem~\ref{fact:strictlocalmaxima} and Lemma~\ref{fact:strictlocalminima} to the case of the stochastic differential equation \eqref{eq:SDE_with_drift}, we can make use of the following representation of the solution $X$ as a time change of Brownian motion:
\begin{thmC}[{\cite[5.5.13]{MR1121940}}]\label{prop:drift_removal}
	Assume Assumption~\ref{cond:SDE_with_drift}.
	Fix a number $c\in\R$ and define the scale function
	\begin{align*}
		p(x)&:=\int_c^x\exp\left\{-2\int_c^\xi\frac{b(\zeta)d\zeta}{\sigma^2(\zeta)}\right\}d\xi;\hspace{10pt}x\in\R
	\end{align*} 
	and inverse $q:(p(-\infty),p(\infty))\rightarrow\R$ of $p$.
	   A process $X=\{X_t\}_{t\geq0}$ is a strong solution of equation \eqref{eq:SDE_with_drift} if and only if the process $Y:=\{Y_t=p(X_t)\}_{t\geq0}$ is a strong solution of
	\begin{align}
		Y_t=Y_0+\int_0^t\tilde{\sigma}(Y_s)dW_s;\hspace{10pt}0\leq t<\infty,\label{eq:transformed_SDE}
	\end{align}
	where
	\begin{gather*}
		p(-\infty)<Y_0<p(\infty)\hspace{10pt}\asure,\\
		\tilde{\sigma}(y)=\begin{cases}
			p'(q(y))\sigma(q(y));&p(\infty)<y<p(\infty),\\
			0;&\text{otherwise}.
		\end{cases}
	\end{gather*}
\end{thmC}

\begin{fact}\label{fact:martingale_theo}
Let $(\Omega,\F,\Prob)$ be a probability space.
\begin{enumerate}[(i)]
    \item In equation \eqref{eq:transformed_SDE}, the Brownian motion $\{W_t\}_{t\geq0}$ is known to be a \emph{continuous local martingale}. The definition of a continuous local martingale can be found in 1.5.15 of the reference \cite{MR1121940}.
    Furthermore, the stochastic integral $\{\int_0^t\tilde{\sigma}(Y_s)dW_s\}_{t\geq0}$ is also a continuous local martingale, provided that the condition
\begin{align*}
\Prob\left(\omega; \int_0^t\tilde{\sigma}^2(Y_s(\omega))ds<\infty\right) = 1 \quad \text{for every } t \in [0,\infty)
\end{align*}
is satisfied.
This fact is also mentioned in Section 3.2.D of \cite{MR1121940}.
    \item  The stochastic process $\{\int_{0}^{t}\tilde{\sigma}^2(Y_s)ds\}_{t\geq0}$ is indeed referred to as the \emph{quadratic variation} of the continuous local martingale $\{\int_0^t\tilde{\sigma}(Y_s)dW_s\}_{t\geq0}$.
    The definition of quadratic variation can be found in 1.5.18 of the reference \cite{MR1121940}.
    This fact is also mentioned in Section 3.2.D of \cite{MR1121940}.
    \item  Let $M:=\{M_t\}_{t\geq0}$ be a continuous local martingale starting at zero and $\lrangle{M}:=\{\lrangle{M}_t\}_{t\geq0}$ be the quadratic variation of $M$.
    Consider a probability space $(\hat{\Omega},\hat{\F},\hat{\Prob})$ in which a Brownian motion exists.
    According to the results in 3.4.6 and 3.4.7 of the reference \cite{MR1121940}, there exists a Brownian motion $\Tilde{B}:=\{\Tilde{B}_t\}_{t\geq0}$ defined on the probability space $(\tilde{\Omega},\tilde{\F},\tilde{\Prob})=(\Omega\times\hat{\Omega},\F\otimes\hat{\F},\Prob\otimes\hat{\Prob})$ such that
\begin{align*}
    M_t(\omega)=B_{\lrangle{M}_t}(\omega,\Tilde{\omega}) \quad \text{almost surely } \tilde{\Prob}.
\end{align*}
This provides a representation of the continuous local martingale $M$ in terms of the Brownian motion $\Tilde{B}$.
\end{enumerate}
\end{fact}

The following lemma gives a representation of SDE \eqref{eq:SDE_with_drift} by time change of Brownian motion.

\begin{thm}\label{theo:rep_SDE_time_change}
	Suppose that $\sigma:\R\rightarrow\R$ and $b:\R\rightarrow\R$ satisfies Assumption~\ref{cond:SDE_with_drift}.
Then, a unique, strong solution exists \eqref{eq:SDE_with_drift}.
Moreover, there exist
\begin{itemize}
    \item a probability space $(\hat{\Omega},\hat{\F},\hat{\Prob})$,
    \item a Brownian motion $B$ and an nonnegative continuous strictly increasing process$\{Z_t\}_{t\geq0}$ with $Z_0=0$ on $(\tilde{\Omega},\tilde{\F},\tilde{\Prob}):=(\Omega\times\hat{\Omega},\F\otimes\hat{\F},\Prob\otimes\hat{\Prob})$, and
    \item a strictly increasing continuous function $p:\R\rightarrow\R$
\end{itemize}
such that
\begin{align}
	X_t(\omega)=p^{-1}(p(\xi)+B_{Z_t}(\omega,\hat{\omega}))\hspace{10pt}0\leq t<\infty,
\end{align}
a.s. $\tilde{\Prob}$.
Here $\xi=X_0$ is the constant initial value of the SDE \eqref{eq:SDE_with_drift}.
\end{thm}

\begin{proof}
    The existence and uniqueness of the solution follow from Proposition~\ref{prop:strong_sol_of_SDE}.
 Since $\sigma$ and $b$ satisfy Assumption~\ref{cond:SDE_with_drift}, we can deduce from Theorem~\ref{prop:drift_removal} that there exists a continuous injection $p:\R\rightarrow\R$ and a continuous function $\tilde{\sigma}:\R\rightarrow\R$ such that $\{Y_t(\omega)\}_{t\geq0}=\{p(X_t(\omega))\}_{t\geq0}$ is a strong solution of the SDE \eqref{eq:transformed_SDE}.
 From the Definition~\ref{defi:SDE_with_drift} of the strong solution, $\Prob[\int_0^t\tilde{\sigma}^2(Y_s(\omega))ds<\infty]=1$ for every $t$.
 Then, by (i) of Fact~\ref{fact:martingale_theo}, we know that $Y_t(\omega)-Y_0(\omega)$ is a continuous local martingale starting at zero.
 From (ii) of Fact~\ref{fact:martingale_theo}, the quadratic variation $\lrangle{Y}$ of $Y$ is given by $\{\int_{0}^{t}\tilde{\sigma}^2(Y_s)ds\}_{t\geq0}$.
 Moreover, from the definition $Y_t(\omega)=p(X_t(\omega))$ and the construction of $\tilde{\sigma}$ in \eqref{eq:transformed_SDE}, we have $\tilde{\sigma}^2(Y_t(\omega))>0$ for every $t\in[0,\infty)$. Therefore, $t\mapsto\lrangle{Y}_t(\omega)$ is strictly increasing almost surely.
 Now, by (iii) of Fact~\ref{fact:martingale_theo}, there exists a probability space $(\hat{\Omega},\hat{\F},\hat{\Prob})$ such that there exists a Brownian motion $B$ on $(\tilde{\Omega},\tilde{\F},\tilde{\Prob})=(\Omega\times\hat{\Omega},\F\otimes\hat{\F},\Prob\otimes\hat{\Prob})$ such that
\begin{align*}
Y_t(\omega)=Y_0(\omega)+B_{\lrangle{Y}_t}(\omega,\hat{\omega}) \quad \text{almost surely on } \tilde{\Prob}.
\end{align*}
We obtain the desired result by setting $Z_t:=\lrangle{Y}_t$.
\end{proof}

\begin{proof}[Proof of Lemma~\ref{lemm:local_extreme}]
    From Theorem~\ref{theo:rep_SDE_time_change}, we have $X_t(\omega)=p^{-1}(p(\xi)+B_{Z_t}(\omega,\hat{\omega}))$ for all $t\in[0,\infty)$ almost surely in the extended probability space $(\tilde{\Omega},\tilde{\F},\tilde{\Prob})=(\Omega\times\hat{\Omega},\F\otimes\hat{\F},\Prob\otimes\hat{\Prob})$, where $\{B_t\}_{t\geq0}$ is a Brownian motion and $\{Z_t\}_{t\geq0}$ is a continuous strictly increasing process.
    \begin{enumerate}[(i)]
        \item
    Set
    \begin{align*}
    \mathcal{L}^{p(a)-p(\xi)}{(\omega,\hat{\omega})}:=\{t\geq0;B_{Z_t}(\omega,\hat{\omega})=p(a)-p(\xi)\},\hspace{10pt}a\in\R,\ (\omega,\hat{\omega})\in\Omega\times\hat{\Omega}.
    \end{align*}
    Since $t\mapsto Z_t(\omega,\hat{\omega})$ is strictly increasing and continuous almost surely $\tilde{\Prob}$, Theorem~\ref{fact:strictlocalmaxima}-(i) implies that $\mathcal{L}^{p(a)}{(\omega,\hat{\omega})}$ is dense in itself almost surely $\tilde{\Prob}$.
    Since $X_t(\omega)=p^{-1}(p(\xi)+B_{Z_t}(\omega,\hat{\omega}))$ almost surely $\tilde{\Prob}$, the set $\mathcal{L}^{a}_{\omega}:=\{t\geq0;X_t(\omega)=a\}$ is dense in itself almost surely $\Prob$.
    
        \item Since $p^{-1}$ is strictly increasing, the points of local maximum and local minimum of $X$ are the same as those of $\{B_{Z_t}\}_{t\geq0}$.
    Since $t\mapsto Z_t(\omega,\hat{\omega})$ is strictly increasing $\tilde{\Prob}$-almost surely, every point of local maximum and local minimum of $t\mapsto B_{Z_t}(\omega,\hat{\omega})$ is strict.
    Since $X_t(\omega)=p^{-1}(B_{Z_t}(\omega,\hat{\omega}))$ almost surely $\tilde{\Prob}$, every point of local maximum and local minimum of $t\mapsto X_t$ is strict almost surely $\Prob$.
    \qedhere
    \end{enumerate}
\end{proof}

\end{document}